\documentclass{article}
\usepackage[margin=1in]{geometry}
\usepackage{color, colortbl}
\usepackage[dvipsnames]{xcolor}
\usepackage{graphicx}
\usepackage{authblk}
\usepackage{float}
\usepackage{amsmath,amssymb,amsthm}
\usepackage{braket}
\usepackage{bbm}
\usepackage[colorlinks=true,linkcolor=red,bookmarks=true,breaklinks=true]{hyperref}
\usepackage[backend=biber,style=numeric-comp,sorting=none]{biblatex}
\usepackage{multirow}

\newtheorem{thm}{\protect\theoremname}
    \theoremstyle{plain}
\newtheorem{remark}[thm]{Remark}
    \theoremstyle{plain}
\newtheorem*{lem*}{\protect\lemmaname}
    \theoremstyle{plain}

    \theoremstyle{plain}
\newtheorem{theorem}[thm]{Theorem}

\newtheorem{lemma}[thm]{Lemma}

\newtheorem{result}[thm]{Result}

\title{
Quantum Differential Equation Solvers with Low State Preparation Cost: Eliminating the Time Dependence in Dissipative Equations}
\author[1,2]{Gengzhi Yang}
\author[3]{Akwum Onwunta}
\author[4,$\dag$]{Dong An}
\affil[1]{Joint Center for Quantum Information and Computer Science, University of Maryland, College Park}
\affil[2]{Department of Mathematics, University of Maryland, College Park}
\affil[3]{Department of Industrial and Systems Engineering,
Lehigh University}
\affil[4]{Beijing International Center for Mathematical Research (BICMR), Peking University}
\affil[$\dag$]{\href{mailto:dongan@pku.edu.cn}{dongan@pku.edu.cn}}
\date{}

\begin{document}
\maketitle
\begin{abstract}
    Linear dissipative differential equation is a fundamental model for a large number of physical systems, such as quantum dynamics with non-Hermitian Hamiltonian, open quantum system dynamics, diffusion process and damped system. 
    In this work, we propose efficient quantum algorithms for simulating linear dissipative differential equations. 
    The key idea of our algorithms is to perform the simulation only over an effective time period when the dynamics has not significantly dissipated yet, rather than over the entire physical evolution period. 
    We conduct detailed analysis on the complexity of our algorithms and show that, while maintaining low state preparation cost, our algorithms can completely eliminate the time dependence. 
    This is a more than exponential improvement compared to the previous state-of-the-art quantum algorithms. 
\end{abstract}

\tableofcontents

\section{Introduction}
Differential equations serve as an important tool to represent real-world problems mathematically. 
It is of great hope that quantum computers can help with the process of solving differential equations.
Let us consider a system of linear ordinary differential equations (ODEs) of the form
\begin{equation}
\label{eqn:main-ode}
    \frac{\mathrm{d}u(t)}{\mathrm{d}t} = A(t) u(t) + b(t), \quad 
    u(0) = u_0,
\end{equation}
where $A(t) \in \mathbb{C}^{N\times N}$, $b(t) \in \mathbb{C}^N$.  
A general solution to this ODE would be
\begin{equation}
    \label{eqn:general-solution}
    u(T) = \mathcal{T}e^{\int_0^T A(s)\mathrm{d}s}u(0) + \int_0^T \mathcal{T}e^{\int_t^T A(s)\mathrm{d}s}b(t)\mathrm{d}t.
\end{equation}
When $N$ is exponentially large, the classical computers are not capable of representing the solution vector, not to mention solving the entire equation. 
Fortunately, quantum computers can efficiently encode high-dimensional vectors in the amplitudes of quantum states, thus potentially enable efficient simulation of the equation. 

As the development of efficient quantum linear system algorithms (QLSAs)~\cite{harrow2009quantum,ambainis2012variable,Childs2017quantum,Subasi2019Quantum,an2022quantum,costa2022optimal,dalzell2024shortcut,low2024eigen}, significant
progress on solving ODEs has been made based on QLSAs~\cite{berry2014high,berry2017quantum,berry2024quantum,liu2021efficient,childs2020quantum,childs2021high,krovi2023improved,low2024quantum}. 
These algorithms first discretize the time variable in the ODE, then formulate the discretized ODE as a linear system of equations and apply QLSA. 
The quantum ODE solvers based on QLSA generally do not attain an optimal state preparation cost, other than the recently developed block-preconditioned technique~\cite{low2024quantum}. 
Besides, recent breakthroughs, such as time-marching method~\cite{fang2023time}, linear combination of Hamiltonian simulation (LCHS)~\cite{an2023linear,an2023quantum}, Schr\"{o}dingerisation~\cite{jin2022quantum,jin2025schrodingerizationmethodlinearnonunitary}, quantum eigenvalue processing~\cite{low2024eigen}, moment-matching dilation framework~\cite{li2025linear}, and methods related to Lindbladian simulation~\cite{shang2024designnearlyoptimalquantum,fang2025qubitefficientquantumalgorithmlinear}, allow us to solve the linear ODEs by directly implementing the time-evolution operator without invoking the QLSAs. 
Remarkably, this type of evolution-based solvers has low state preparation cost, and many of them simultaneously achieve near-optimal overall query complexity which is almost linear in the evolution time $T$~\cite{an2023quantum,low2024eigen,jin2025schrodingerizationmethodlinearnonunitary,shang2024designnearlyoptimalquantum}. 

Although the quantum differential equation solvers have found a wide range of applications~\cite{an2021quantum,clayton2024differentiable, liu2024toward, jin2023quantum,liu2024towards},
in general, the almost linear time dependence in generic quantum ODE algorithms cannot be significantly improved due to the no-go theorem in Hamiltonian simulation problems~\cite{berry2014exponential,gilyen2019quantum,kieferova2019simulating}. 
However, it is still possible and practically relevant 
to look for specific cases and design tailored quantum algorithms with higher efficiency. 
One specific family of such linear ODEs is \emph{dissipative ODEs}, whose time-evolution operator $\mathcal{T}e^{\int_0^t A(s)\mathrm{d}s}$ would drop exponentially in time. 
As discussed in~\cite{jennings2024cost,an2024fast}, dissipative ODEs can model a large number of real-world physical systems, including quantum dynamics with non-Hermitian Hamiltonian, open quantum system dynamics, diffusion process, plasma system, damped oscillator system, and dissipative weakly nonlinear ODEs. 
It has been demonstrated in~\cite{jennings2024cost,an2024fast} that we can fast-forwardably solve this type of ODEs in sub-linear time without adopting further techniques, indicating that the fast-forwarding property is an intrinsic property of dissipative ODEs. 
Nevertheless, previous works~\cite{jennings2024cost,an2024fast} only study QLSA-based algorithms with high state preparation cost, and their overall query complexity still scale polynomially in time. 

In this work, we investigate
to what extent we can fast-forward the process of simulating dissipative ODEs. 
Our goal is to prepare the final state proportional to $u(T)$
or the history state proportional to $\sum_{j=0}^{M-1}\ket{j}\|u(jT/M)\|\ket{u(jT/M)}$ within the error tolerance $\epsilon$, where $M$ is the number of time steps. 
We introduce new quantum algorithms based on the time-marching and LCHS methods. 
For dissipative ODEs, our new algorithms can completely eliminate the time dependence, resulting in an $\mathcal{O}(1)$ overall query complexity in terms of the evolution time $T$, while maintaining low state preparation cost. 

Our algorithms are based on a simple yet powerful observation. 
Since the time evolution operator $\mathcal{T}e^{\int_0^t A(s)\mathrm{d}s}$ of dissipative ODEs exhibits an exponential decay in time, for large $T$, the initial condition $u(0)$ and the inhomogeneous term $b(t)$ imposed at time far away from $T$ do not have significant influence on the solution state $\ket{u(T)}$. 
Therefore, we design our algorithms by applying time-marching or LCHS only on the time period $[T-T_0,T]$ for an effective simulation time $T_0$ which is dependent on the dissipative rate of $A(t)$ but independent of $T$. 
As a result, the query complexity of our algorithms is also independent of $T$. 

Compared with the existing results, our algorithms significantly improve the efficiency of long-time simulation of dissipative ODEs. 
While the best existing result of the final state preparation for dissipative ODEs has query complexity $\widetilde{\mathcal{O}}(T^{1/2})$~\cite{an2024fast}, our results do not depend on $T$ at all, achieving a more than exponential speedup. 
Furthermore, our results fill in the blanks on preparing the history state of ODEs (even only semi-dissipative) with low state preparation cost.

\section{Quantum ODE solvers with low state preparation cost}
We first briefly review two quantum ODE solvers on which our algorithms will be built: the time-marching algorithm~\cite{fang2023time} and the LCHS algorithm~\cite{an2023quantum}. Both algorithms directly implement a block-encoding of the time evolution operator $\mathcal{T}e^{\int_{0}^{T}A(s)\mathrm{d}s}$, and thus have low query complexity to the initial state preparation oracle. 

The time-marching algorithm is an analog of classical time-marching strategy, where the numerical solution is advanced in small discrete time steps. 
Specifically, the algorithm implements truncated Dyson series to approximate the short-time evolution operator $\mathcal{T}e^{\int_{t}^{t+\Delta t}A(s)\mathrm{d}s}$ for a short time step size $\Delta t$ and multiplies them together. 
While a naive multiplication of the short-time evolution operators leads to an exponentially small success probability, the algorithm resolves this issue by uniformly amplifying the success probability at each step, which can be done by using quantum singular value transformation~\cite{gilyen2019quantum}. 
The overall query complexity to the input model of $A(t)$ scales $\mathcal{O}(T^2\log^2(T/\epsilon))$. 

The LCHS algorithm is based on an important observation that the operator $\mathcal{T}e^{\int_{0}^{t}A(s)\mathrm{d}s}$ can be represented as a continuous weighted sum of a set of Hamiltonian simulation problems. 
Specifically, for a coefficient matrix $A(t) = L(t) + i H(t)$ where $L(t)$ and $H(t)$ are the Hermitian and anti-Hermitian parts of $A(t)$, we have 
\begin{equation}\label{eqn:LCHS}
    \mathcal{T}e^{\int_{0}^{t}A(s)\mathrm{d}s} = \int_{\mathbb{R}}  \frac{1}{2\pi e^{-2^{\beta}} (1-ik) e^{(1+ik)^{\beta}}} U_k(t) \mathrm{d}k, 
\end{equation}
where $\beta \in (0,1)$ and $U_k(t) = \mathcal{T} e^{i \int_0^t (kL(s) + H(s) ) \mathrm{d} s }$ solves the Hamiltonian simulation problem with Hamiltonian $kL(s) + H(s)$. 
Then, the LCHS algorithm discretizes the integral in Equation~\eqref{eqn:LCHS}, implements each unitary operator $U_k(t)$ by the truncated Dyson series method~\cite{Berry_2015,low2018hamiltonian}, and performs linear combination by the quantum linear combination of unitaries (LCU) technique~\cite{Childs2012Hamiltonian,Childs2017quantum}. 
Its overall query complexity to the input model of $A(t)$ scales $\mathcal{O}(T \log^{2+o(1)} (T/\epsilon) )$, which is near-optimal on both $T$ and $\epsilon$. 

\section{Quantum algorithms for dissipative ODEs}

\subsection{Dissipative ODEs}
The time-marching algorithm and the LCHS algorithm are time-efficient for ODEs with certain stability condition, namely $A(t) + A^{\dagger}(t) \leq 0$. 
Such ODEs are sometimes referred to as semi-dissipative ODEs. 
For fast-forwarding, following~\cite{an2024fast}, we consider (strictly) dissipative ODEs defined as Equation~\eqref{eqn:main-ode} with 
\begin{equation}
    A(t) + A^\dag(t) \leq -2\eta <0
\end{equation}
for some positive real constant $\eta$. 
Notice that dissipative ODEs exclude all Hamiltonian simulation problems, and thus solving dissipative ODEs with fast-forwarded sublinear scaling in $T$ does not conflict with the no-go theorem in the Hamiltonian simulation case. 

A key feature of dissipative ODEs is the fast decay of the time-evolution operator, quantitatively described in the following lemma. 
\begin{lemma}[Lemma 5~\cite{an2024fast}]
\label{lemma:dissi-bound}
    For any $0 \leq t_0 \leq t_1 \leq T$, if $A(t) + A^\dag(t) \leq -2\eta <0$, we have
    \begin{equation}
        \label{eqn:dissi-upper-bound}
        \left\|\mathcal{T}e^{\int_{t_0}^{t_1}A(s)\mathrm{d}s}\right\|        
        \leq
        e^{-\eta(t_1 - t_0)}.
    \end{equation}
\end{lemma}

Lemma~\ref{lemma:dissi-bound} shows that the time-evolution operator of dissipative ODEs decays exponentially in time. 
In this setting, if the dissipative ODE system is homogeneous (i.e., $b(t) \equiv 0$), then the final state preparation for sufficiently long time is trivial, since the norm of $u(T) = \mathcal{T}e^{\int_{0}^{T}A(s)\mathrm{d}s} u_0$ decays exponentially in $T$ and we can simply use $0$ to approximate the final solution. 
So we always only consider the inhomogeneous ODEs when we discuss the final state preparation in the dissipative case.

\subsection{Quantum algorithms}

Now we describe how to design quantum algorithms for dissipative ODEs with fast-forwarded time dependence. 
As discussed in the previous section, the long-time effect of $u_0$ and $b(t)$ on dissipative ODEs is almost negligible. 
As illustrated in Figure~\ref{fig:nd}, instead of simulating the whole time interval $[0,T]$, we may only simulate the dynamics up to some effective simulation time $T_0$ to get a sufficiently accurate approximation. 
According to Lemma~\ref{lemma:dissi-bound}, we can choose the effective simulation time $T_0 = (1/\eta)\log(1/\epsilon)$, which does not depend on $T$ any more, to bound the truncation of long-time dissipation by $\epsilon$. 

\begin{figure}[ht]
    \centering
    \includegraphics[width=0.49\textwidth]{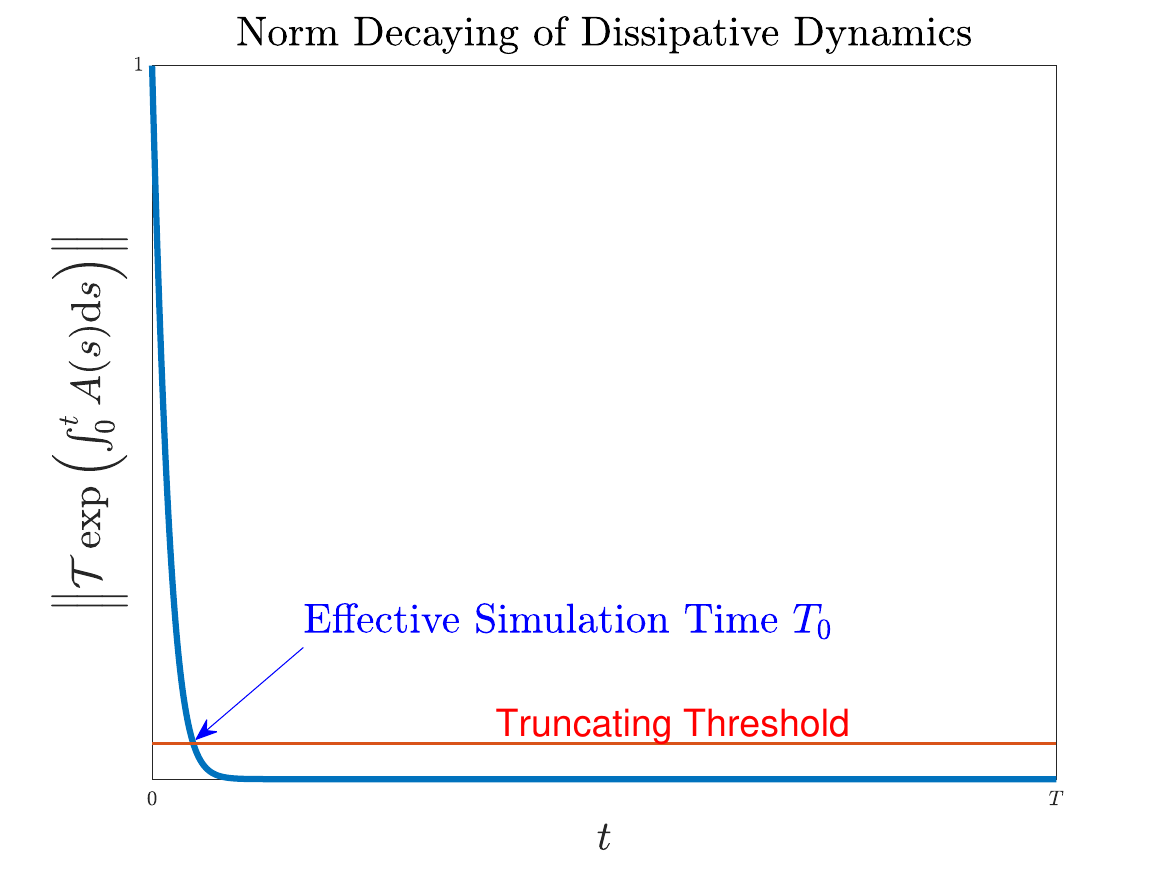} 
    \includegraphics[width=0.49\textwidth]{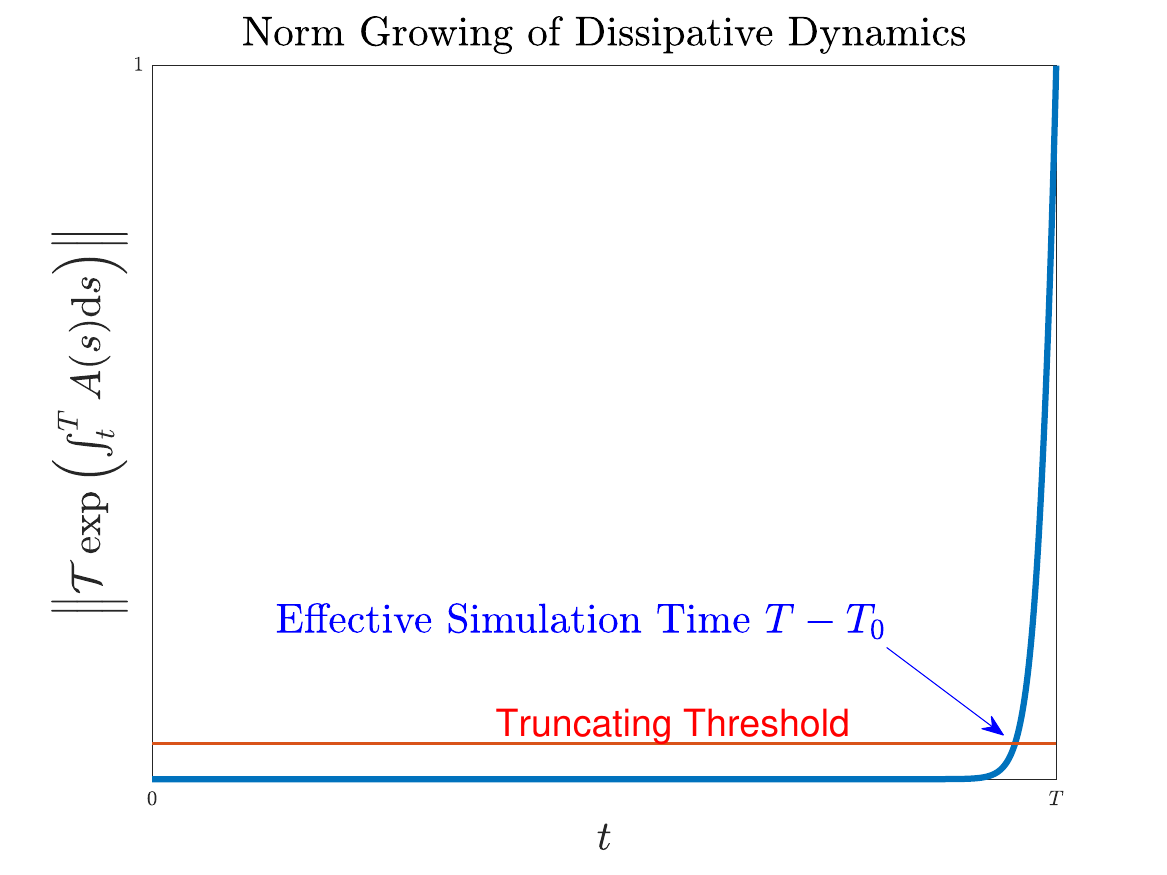}
    \caption{Examples of $\displaystyle\left\|\mathcal{T}e^{\int_0^t A(s) \mathrm{d}s}\right\|$ and $\displaystyle\left\|\mathcal{T}e^{\int_t^T A(s) \mathrm{d}s}\right\|$.}
    \label{fig:nd}
\end{figure}

Based on this observation, we have an accurate approximation of the solution $u(T)$ at time $T \geq T_0$ from Equation~\eqref{eqn:general-solution} being 
\begin{equation}\label{eqn:dissipative-solution-truncated}
    u(T) \approx \int_{T-T_0}^T \mathcal{T}e^{\int_t^T A(s)\mathrm{d}s} b(t)\mathrm{d}t. 
\end{equation}
Our quantum algorithm for final state preparation can be constructed from Equation~\eqref{eqn:dissipative-solution-truncated}. 
We first discretize the integral in Equation~\eqref{eqn:dissipative-solution-truncated} by Gaussian quadrature, resulting in a linear combination of $\mathcal{T}e^{\int_{t_j}^T A(s)\mathrm{d}s} b(t_j)$ at different times $t_j \in [T-T_0, T]$. 
Then, we apply the time-marching method or the LCHS algorithm to construct each $\mathcal{T}e^{\int_{t_j}^T A(s)\mathrm{d}s} b(t_j)$, and perform the linear combination step by LCU.

For history state preparation, the goal is to prepare 
a quantum state proportional to $\sum_{k=0}^{M-1}\ket{k}\|u(kh)\|\ket{u(kh)}$ for a step size $h$ and number of time steps $M = T/h$. 
In the homogeneous case, as the solution beyond $t = T_0$ almost vanishes, we only need to prepare $\sum_{k=0}^{M_0-1}\ket{k}\|u(kh)\|\ket{u(kh)}$ where $M_0 = T_0/h$, and this can be constructed by applying the time-marching or LCHS algorithm to time $kh$ controlled by the register $\ket{k}$ encoding the index. 
In the inhomogeneous case, we may prepare $\sum_{k=0}^{M-1}\ket{k}\|u(kh)\|\ket{u(kh)}$ by the same methodology as preparing the final state at time $kh$, controlled by the register $\ket{k}$. 
Notice that for small $k$ such that $kh \leq T_0$, we simply compute $\ket{u(kh)}$ for the entire time interval $[0,kh]$. 

We remark that if we want to implement our algorithm, we should know an \emph{a priori} estimation of $\eta$ with constant multiplicative error. 
Clearly estimating $\eta$ is equivalent to estimating the ground energy of the Hermitian $-(A(t) + A^\dag(t))$, which is QMA-hard if no assumptions are included~\cite{aharonov2002quantum,aharonov2009power,kempe2006complexity,kitaev2002classical,oliveira2005complexity}. 
Nevertheless, in some practical applications, the value of $\eta$ can be obtained theoretically. 
Besides, various efficient quantum algorithms have been developed~\cite{dong2022ground,lin2020near,lin2022heisenberg} if we have the access a quantum state that has a large overlap with the ground state. 
Even more, this condition can be relaxed by simulating Lindbladian dynamics~\cite{ding2024single, lin2025dissipative, ding2025end}, where under certain circumstances the mixing time is short~\cite{zhan2025rapid}.

\subsection{Complexity}
Here we discuss the query complexity of our quantum algorithms for dissipative ODEs. 
We first state the query models, which mainly follow~\cite{berry2024quantum}. 
For the coefficient matrix $A(t)$, we assume its time-dependent block-encoding $U_A$ such that 
\begin{equation}\label{eqn:def_oracle_At}
    U_A \ket{t} \ket{0}_a\ket{\psi} =  \ket{t} \left( \ket{0}_a \frac{A(t)}{\alpha_A}\ket{\psi} + \ket{\perp_a}\right). 
\end{equation}
Here $\alpha_A \geq \max_t \|A(t)\|$ is the block-encoding normalization factor, and $\ket{\perp_a}$ denotes a possibly unnormalized state such that $(\ket{0}_a\bra{0}_a \otimes I ) \ket{\perp_a} = 0$. 
$\ket{t}$ represents an encoding of time $t$. 
The actual encoding may vary under different scenarios, resulting in variants of the HAM-T model (see Appendix~\ref{section:input-models} for more details). 
For the inhomogeneous term $b(t)$, we assume its time-dependent state preparation oracle $U_b$ such that 
\begin{equation}\label{eqn:def_oracle_b}
    O_b \ket{t} \ket{0^n} = \ket{t} \ket{b(t)}. 
\end{equation}
Additionally, we are given an oracle $O_u$ that prepares the initial state as 
\begin{equation}
    O_u \ket{0^n} = \ket{u(0)}.
\end{equation}

The complexity of our algorithms for dissipative ODEs is given in the following result. 
We mainly focus on the scaling in terms of the simulation time $T$ and the error tolerance $\epsilon$ here, 
while other dependencies are shown in Appendix~\ref{sec:tm} and Appendix~\ref{sec:lchs}. 

\begin{result}
    \label{result:dissi}
    We are able to solve Equation~\eqref{eqn:main-ode} with $A(t) + A^\dag(t) \leq -2\eta <0$ for some constant $\eta$,
    by preparing an $\epsilon$-approximation of the history state/final state with only
    \begin{equation}
        \mathcal{O}\left(\log^3(1/\epsilon)\right)
    \end{equation}
    queries to matrix $A(t)$ using the time-marching method, or 
    \begin{equation}
        \mathcal{O}\left(\log^{3+o(1)}(1/\epsilon)\right)
    \end{equation}
    queries to matrix $A(t)$ using the LCHS method, while maintaining a low state preparation cost independent of $T$ and $\epsilon$. 
\end{result}

As discussed earlier, the main idea behind Result~\ref{result:dissi} is that we only need to simulate the dynamics in Equation~\eqref{eqn:main-ode} for a short time if $A(t) + A^\dag(t) \leq -2\eta < 0$ holds true for some positive constant $\eta$. 
Specifically, for constant $\eta$, the effective simulation time $T_0$ is $\mathcal{O}(\log(1/\epsilon))$, independent of the physical time $T$, resulting in $\mathcal{O}(\text{poly}\log(1/\epsilon))$ overall query complexity which is also independent of $T$. 
We remark an exceptional case which is the history state with an inhomogeneous term: even though we only need to 
simulate for $\mathcal{O}(\log(1/\epsilon))$ time for every time step, there are $\mathcal{O}(T)$ time steps in total. 
Fortunately, the oracle $O_b$ defined in Equation~\eqref{eqn:def_oracle_b}, which can be viewed as a controlled time-dependent state preparation oracle, could resolve this problem since the construction for every time step is now superposed. 
Such an oracle has been widely used in previous quantum ODE algorithm design~\cite{low2018hamiltonian,fang2023time,an2023quantum,berry2024quantum}, although its concrete construction might require more resources. 

\begin{table}[H]
    \renewcommand{\arraystretch}{1.8}
    \centering
    \scalebox{0.83}{
    \begin{tabular}{c|c|c|c|c}\hline\hline
        \multirow{3}{4em}{\textbf{Methods} } & \multicolumn{4}{c}{\textbf{Query complexities to the coefficient matrix}} \\
        \cline{2-5} & \multicolumn{2}{c|}{ History state } & \multicolumn{2}{c}{ Final state } \\
        \cline{2-5} & Semi-dissipative & Dissipative & Semi-dissipative & Dissipative \\ \hline
       QLSP (Homogeneous) & $\mathcal{O} (T (\log(T/\epsilon))^2) $ & $\mathcal{O}( (\log(1/\epsilon))^2) $ & $\mathcal{O} (T (\log(T/\epsilon))^2) $ & Not applicable \\\hline
       QLSP (Inhomogeneous)  & $\mathcal{O} (T (\log(T/\epsilon))^2) $ & $\mathcal{O}( \log(T)(\log(1/\epsilon))^2) $ & $\mathcal{O} (T (\log(T/\epsilon))^2) $ & $\mathcal{O} (\sqrt{T} (\log(T/\epsilon))^2) $ \\\hline
       Time-marching & \textcolor{blue}{ $ \mathcal{O} (T^2 (\log(T/\epsilon))^2) $} & \textcolor{red}{$\mathcal{O}( (\log(1/\epsilon))^3) $} & $\mathcal{O} (T^2 (\log(T/\epsilon))^2) $ & \textcolor{red}{$\mathcal{O}( (\log(1/\epsilon))^3) $} \\\hline
       LCHS & \textcolor{blue}{ $\mathcal{O} (T (\log(T/\epsilon))^{2+o(1)} ) $} & \textcolor{red}{$\mathcal{O}((\log(1/\epsilon))^{3+o(1)} ) $} & $ \mathcal{O}(T (\log(T/\epsilon))^{2+o(1)} ) $ & \textcolor{red}{$\mathcal{O}((\log(1/\epsilon))^{3+o(1)} ) $} \\\hline\hline
    \end{tabular}
    }
    \caption{Query complexities of linear ODE algorithms to the coefficient matrix in terms of evolution time $T$ and error $\epsilon$. Here our fast-forwarding results are highlighted in red, and our additional results for history state preparation are in blue. The QLSP results are from~\cite{berry2024quantum} for semi-dissipative case and~\cite{an2024fast} for dissipative case. For final state preparation of semi-dissipative ODEs, the result of time-marching is from~\cite{fang2023time} (and also re-derived in Appendix~\ref{sec:tm}), and the result of LCHS is from~\cite{an2023quantum}. }
    \label{table:results}
\end{table}

\begin{table}[H]
    \renewcommand{\arraystretch}{1.8}
    \centering
    \scalebox{0.83}{
    \begin{tabular}{c|c|c|c|c}\hline\hline
        \multirow{3}{4em}{\textbf{Methods} } & \multicolumn{4}{c}{\textbf{Query complexities to the coefficient matrix}} \\
        \cline{2-5} & \multicolumn{2}{c|}{ History state } & \multicolumn{2}{c}{ Final state } \\
        \cline{2-5} & Semi-dissipative & Dissipative & Semi-dissipative & Dissipative \\ \hline
       QLSP (Homogeneous) & $\mathcal{O} (T \log(T/\epsilon)) $ & $\mathcal{O}( (\log(1/\epsilon))^2) $ & $\mathcal{O} (T \log(T/\epsilon)) $ & Not applicable \\\hline
       QLSP (Inhomogeneous)  & $\mathcal{O} (T \log(T/\epsilon)) $ & $\mathcal{O}( \log(T)(\log(1/\epsilon))^2) $ & $\mathcal{O} (T \log(T/\epsilon)) $ & $\mathcal{O} (\sqrt{T} (\log(T/\epsilon))^2) $ \\\hline
       Time-marching & \textcolor{blue}{ $ \mathcal{O} ( 1 ) $} & \textcolor{red}{$\mathcal{O}( 1 ) $} & $\mathcal{O} ( 1 ) $ & \textcolor{red}{$\mathcal{O}( 1 ) $} \\\hline
       LCHS & \textcolor{blue}{ $ \mathcal{O} ( 1 ) $} & \textcolor{red}{$\mathcal{O}( 1 ) $} & $\mathcal{O} ( 1 ) $ & \textcolor{red}{$\mathcal{O}( 1 ) $} \\\hline\hline
    \end{tabular}
    }
    \caption{Query complexities of linear ODE algorithms to the initial state in terms of evolution time $T$ and error $\epsilon$. Here our fast-forwarding results are highlighted in red, and our additional results for history state preparation are in blue. Sources of the existing results are the same as described in the caption of Table~\ref{table:results}. }
    \label{table:results_state_prep}
\end{table}

\subsection{Related Works and Comparison}
There have been several works dedicated to explore possibility of fast-forwarding quantum ODE algorithms~\cite{jennings2024cost,atia2017fast,gu2021fast, an2022theory}. 
Table~\ref{table:results} and Table~\ref{table:results_state_prep} show a comparison between our results and existing ones, with or without fast-forwarding. 
Our fast-forwarded results breaks the linear time dependence by restricting ourselves to dissipative ODEs, so below we mainly discuss the comparison between our algorithms and the existing fast-forwarded algorithms. 

Among those, the works~\cite{jennings2024cost, an2024fast} consider dissipative ODEs and analyze the complexity of the linear-system-based approach with truncated Dyson series discretization. 
They achieve fast-forwarded sublinear time dependence by carefully bounding the condition number of the linear system. 
As a comparison, our algorithms directly implement the time evolution operator up to an effective simulation time, leading to lower query complexity to the state preparation oracles. 
Furthermore, as shown in Table~\ref{table:results}, for the inhomogeneous dissipative ODEs, the time dependence in our algorithms is improved to not depending on $T$ at all, while previous results scales $\mathcal{O}(\log(T))$ for history state preparation and $\mathcal{O}(\sqrt{T})$ for final state preparation. 
Therefore, our algorithms offer significant computational speedups in simulating physical systems in practice which can be modeled by dissipative ODEs. 
We discuss the applications of our algorithms to non-Hermitian quantum dynamics and reaction-diffusion process in Appendix~\ref{sec:applications}.

In addition, for final state preparation in~\cite{jennings2024cost, an2024fast}, the linear system needs to be coupled with extra padding lines in order to boost the success probability of obtaining the final solution by post-selection. 
The number of the padding lines should also be carefully chosen according to the dissipation rate, otherwise the fast-forwarded complexity could disappear. 
On the contrary, our algorithm does not require a delicate design, but only an estimation of the simulation time.

\section{Additional results on solving semi-dissipative ODEs}

As a side product, we can also prepare history state proportional to $\sum_{k=0}^{M-1}\ket{k}\|u(kh)\|\ket{u(kh)}$ of semi-dissipative ODEs by time-marching or LCHS. 
This is a complementation of existing works where both time-marching and LCHS were designed only for final state preparation. 
The algorithms are similar to our fast-forwarded ones: we simply apply a controlled version of time-marching or LCHS, but here the simulation time should be the physical time $kh$ for every index $k$. 
The following two results give the query complexities of history state preparation, which maintain on the same leval as the final state preparation case. 

\begin{result}
    \label{result:history}
    We are able to solve Equation~\eqref{eqn:main-ode} 
    by preparing an $\epsilon$-approximation of the history state 
    with
    \begin{equation}
        \mathcal{O}\left(T^2\log^2(T/\epsilon)\right)
    \end{equation}
    queries to the matrix $A(t)$ using the time-marching method, or 
     \begin{equation}
        \mathcal{O}\left(T\log^{2+o(1)}(T/\epsilon)\right)
    \end{equation}
    queries to the matrix $A(t)$ using the time-marching method, 
    while maintaining a low state preparation cost as
    \begin{equation}
        \mathcal{O}\left(1\bigg/\sqrt{\sum_{k=0}^{M-1}\frac{\|u(kh)\|^2/M}{\left(\|u_0\| + \int_0^{kh}|b(t)|\mathrm{d}t\right)^2 }}\right). 
    \end{equation}
\end{result}

\begin{remark}
    Result~\ref{result:history} is stated for the genenal inhomogeneous case. 
    The complexities of homogeneous cases are simply the degenerated version of the result above by setting $b(t) \equiv 0$. 
\end{remark}

\section{Discussion and open questions}
In this work we discuss the complexities of applying time-marching and LCHS on linear ODEs, and especially the dissipative cases. 
Our results improve previous results on preparing the final state of dissipative ODEs from square root in $T$ to not depending on $T$ at all, and fill the blanks of using evolution-based approaches to prepare history states. 
The idea behind fast-forwardly solving dissipative ODEs is to only simulate a short period of time before the state decays exponentially close to zero, which is a simple idea while being lack of consideration in the past.

A drawback of this work is the strong dissipation condition we assume. 
Compared to the standard Lyapunov stability condition $P(t)A(t) + A^\dag(t)P(t) \leq -2\eta <0$ with $P(t) > 0$ which is considered in~\cite{jennings2024cost}, we imposed a stronger version of it, namely $A(t) + A^\dag (t) \leq -2\eta < 0$. 
Since many applications only satisfy the dissipation condition 
in the weaker Lyapunov sense, it is interesting to consider whether we can relax the restriction on $A(t)$ while still maintaining fast-forwarded complexities. 
Besides that, it would also be interesting to find out other conditions that lead to fast-forwarding properties. 

Our work shows that the fast-forwarding property is naturally embedded in the dissipative nature of the system, thus one could expect the same scaling if applying a similar analysis on the linear system solvers. 

\section*{Acknowledgements}

We thank Di Fang, Lin Lin, and Yu Tong for the discussions on the time-marching method. 
DA acknowledges the support by the Innovation Program for Quantum Science and Technology via Project 2024ZD0301900, and the Fundamental Research Funds for the Central Universities, Peking University.

\printbibliography
\clearpage

\appendix

\section{Input models}\label{section:input-models}

In this section, we discuss the input models in our algorithms with more details. 

\textbf{Block-encodings.}
We say $U_A$ is an $(\alpha, a, \epsilon)$ block-encoding of a matrix $A$ if $U_A$ is a unitary and
\begin{equation}
    \left\|(\alpha\bra{0^a}\otimes I) U_A(\ket{0^a}\otimes I) - A\right\| \leq \epsilon.
\end{equation}
Sometimes we do not care about the number of ancilla qubits, then it will be denoted as an $(\alpha, \cdot, \epsilon)$ block-encoding.

\textbf{Coefficient Matrix.} 
The coefficient matrix $A(t)$ in our work is in general time-dependent, so we will assume time-dependent versions of the block-encoding as our input models. 
For final state preparation, following the convention in~\cite{an2023quantum}, assume we are able to query $\text{HAM-T}_{A,q}$ given as 
\begin{equation}
   \left(\bra{0}_{a} \otimes I_n\right) \text{HAM-T}_{A,q} \left(\ket{0}_{a} \otimes I_n \right) = \sum_{s=0}^{M_D - 1}\ket{s}\bra{s}\otimes \frac{A(qh+sh/M_D)}{\alpha_A}.
\end{equation}
Here $I_n$ is the identity matrix of size $2^n \times 2^n$, $\alpha_A$ is a unified normalization factor such that $\alpha_A \geq \sup_t \|A(t)\|$, $M_D$ is the number of time steps for simulating
the dynamics in $[qh, (q+1)h]$, and $a$ is the number of ancilla qubits.
A stronger model required by the history state preparation is 
\begin{equation}
\label{eqn:stronger-input-model}
   \left(\bra{0}_{a} \otimes I_n\right) \text{HAM-T}_{A} \left(\ket{0}_{a} \otimes I_n \right) = \sum_{q}\sum_{s=0}^{M_D - 1}\ket{q}\bra{q} \otimes \ket{s}\bra{s}\otimes \frac{A(qh+sh/M_D)}{\alpha_A},
\end{equation}
where the pre-determined $q$'s are used to represent the time. 

Intuitively, $\text{HAM-T}_{A,q}$ simultaneously block encodes the coefficient matrix $A(t)$ over the time period $[qh, (q+1)h]$ for a single step propagation, and $\text{HAM-T}_{A}$ is a global version of $\text{HAM-T}_{A,q}$. 
Although $\text{HAM-T}_{A}$ is a stronger model than $\text{HAM-T}_{A,q}$, we notice that both of them can be efficiently constructed by one query to the time-dependent block-encoding of $A(t)$ defined in Equation \eqref{eqn:def_oracle_At} with the help of another oracle that maps the index $(q,s)$ to a binary encoding of the corresponding time $qh+sh/M_D$.

\textbf{State Preparation.}
We are given an oracle $O_u$ that prepares the initial state:
\begin{equation}
    O_u \ket{0^n} = \ket{u(0)}.
\end{equation}
To address the inhomogeneous term, we need an oracle $O_b$ that produces $\ket{b(t)}$ 
according to an ancillary register:
\begin{equation}
    O_b \ket{j}\ket{0^n} = \ket{j} \ket{b(t_j)}.
\end{equation}
Here $t_j$'s are uniformly discretized points on $[0,T]$.

For simplicity, we call the oracle $O_{\text{init}}$ as a uniform state preparation oracle, meaning 
\begin{equation}
\label{eqn:init-oracle}
    O_{\text{init}}\ket{0}\ket{0^n} = \ket{0}\ket{u(0)},\quad
    O_{\text{init}}\ket{j}\ket{0^n} = \ket{j}\ket{b(t_j)}.
\end{equation}

\section{Numerical integration}
Throughout this paper, we use the Gaussian Quadrature rule to perform numerical integrations. For a given interval $[a,b]$ and a function $f$ that is smooth enough, there exist weights $\{c_j\}$ and nodes $\{x_j\}$ that satisfy~\cite{an2023quantum}
\begin{equation}
    \left\|\int_a^b f(x) \mathrm{d}x
    - \sum_{j=1}^n c_j f(x_j)\right\| 
    \leq \frac{(b-a)^{2n+1}[n!]^4}{(2n+1)[(2n)!]^3}\max_{\xi \in (a,b)} |f^{(2n)}(\xi)|.
\end{equation}

Treat $\max_{\xi \in (a,b)} |f^{(2n)}(\xi)|$ as a constant,
for an interval of length $T$ that is divided into $M_I$ equi-length sub-intervals, given an error $\epsilon$, we need to restrict that on each interval our integration admits at most an error of $\frac{\epsilon}{M_I}$,
indicating 
\begin{equation}
    \frac{(T/M_I)^{2n+1} [n!]^4}{(2n+1)[(2n)!]^3} \leq \frac{\epsilon}{M_I}.
\end{equation}
From the Stirling's formula, we know that
\begin{equation}
    n! \sim \sqrt{2\pi n}\left(\frac{n}{e}\right)^n.
\end{equation}
The estimation is equivalent to (in the big O notation sense)
\begin{equation}
    \frac{T^{2n+1}n^{4n+2}e^{6n}}{M_I^{2n}e^{4n}2^{6n}n^{6n + 3/2}} \leq \epsilon.
\end{equation}
It is sufficient to take $M_I = \mathcal{O}(\alpha_A T)$ and $n = \mathcal{O}(\log(\alpha_A T/\epsilon))$. The choice of $M_I$ is intended for the construction of Dyson series.

If we are performing the integration for an $N$-dimensional vector, in order to achieve $\epsilon_v$-accuracy in the standard 2-norm, the $\epsilon$ should be chosen as $\epsilon_v/\sqrt{N}$. However this does not affect the complexities of our quantum algorithms because the vectors are always normalized.

\section{Time-Marching method}
\label{sec:tm}

The Time-Marching method was originally proposed in~\cite{fang2023time}, as the first quantum differential equation solver for preparing final states of homogeneous ODEs with optimal state preparation cost. 
However, there might be a mistake in the complexity analysis, and the query complexity in terms of the error is quadratically underestimated. 
For completeness, here we provide a self-contained complexity analysis for the time-marching method in final state preparation for homogeneous ODE. 
Furthermore, we analyze its complexity in several other scenarios which are not discussed in existing literature, including final state preparation for inhomogeneous ODEs, history state preparation, and dissipative ODEs. 

\textbf{One-step Propagator.}
Provided with a choice $\alpha_Ah = \mathcal{O}(1)$, we can construct an $(\alpha_\ell, a_1, \epsilon_0)$-block-encoding $P_\ell$ of $\mathcal{T}e^{\left(\int_{\ell h}^{(\ell+1)h}A(s)\mathrm{d}s\right)}$ with $\left\|\mathcal{T}e^{\left(\int_{\ell h}^{(\ell+1)h}A(s)\mathrm{d}s\right)}\right\| \leq \alpha_\ell = \mathcal{O}(1)$ using $\mathcal{O}\left(\frac{\log(1/\epsilon_0)}{\log\log(1/\epsilon_0)}\right)$ queries to $\text{HAM-T}_{A,\ell}$ by leveraging the technique in~\cite[Theorem 3]{low2018hamiltonian},
but without the uniform singular value amplification.
For the convenience of further analysis, we also write
\begin{equation}
    \alpha_\ell \bra{0^{a_1}}P_\ell\ket{0^{a_1}}
    = \mathcal{T}e^{\int_{\ell h}^{(\ell + 1) h}A(s)\mathrm{d}s} + \epsilon_0\Lambda_\ell,
\end{equation}
where $\|\Lambda_\ell\| \leq 1$ is used to represent the error part of each block-encoding. 

\subsection{Final state preparation}
\subsubsection{Homogeneous case}

For a homogeneous ODE, say
\begin{equation}
\begin{split}
    \frac{\mathrm{d}u(t)}{\mathrm{d}t} &= A(t) u(t),\\
    u(0) &= u_0.
\end{split}
\end{equation}
The idea is to apply a sequence of $\{P_\ell\}$'s 
to get the desired final state.
The issue we need to deal with is the exponentially decreasing success probability, if we directly multiply these block-encodings together. 
Such an issue was tackled by using an amplitude amplifying technique called Uniform Amplitude Amplification~\cite{fang2023time, gilyen2019quantum}. 

\begin{theorem}
\label{thm:uniform-amplitude-amplification}
Let $U$ be an $(\alpha,m,\epsilon_{\text{ori}})$-block-encoding of $\Xi$. 
Then we can construct an~$\left(\alpha', m+1, \epsilon_a\|\Xi\| + \epsilon_{\text{ori}} + \epsilon_a\epsilon_{\text{ori}}\right)$-block-encoding $\tilde{U}$ of $\Xi$ such that $\alpha' \leq \frac{\|\Xi\|}{1-\delta}$,
using $\mathcal{O}\left(\frac{\alpha}{\delta \|\Xi\|}\log\left(\frac{\alpha}{\|\Xi\|\epsilon_a}\right)\right)$ applications of (controlled)-$U$ and its inverse.
\end{theorem}
\begin{proof}
    Consider $U$ as a $(\alpha, m, 0)$-block-encoding of $\Xi + \epsilon_{\text{ori}}\Lambda$ where $\|\Lambda\| \leq 1$. Then, according to~\cite{fang2023time, gilyen2019quantum}, we can construct a
    $\left(\frac{\|\Xi + \epsilon_{\text{ori}}\Lambda\|}{1-\delta/2}, m+1, \epsilon_a\|\Xi + \epsilon_{\text{ori}}\Lambda\|\right)$-block-encoding of $\Xi + \epsilon_{\text{ori}}\Lambda$
    in cost 
    \begin{equation}
        \mathcal{O}\left(\frac{\alpha}{\delta \|\Xi + \epsilon_{\text{ori}}\Lambda\|}\log(\frac{\alpha}{\epsilon_a\|\Xi + \epsilon_{\text{ori}}\Lambda\|})\right) \leq \mathcal{O}\left(\frac{\alpha}{\delta \|\Xi \|}\log(\frac{\alpha}{\epsilon_a\|\Xi \|})\right), 
    \end{equation}
   and this can also be viewed as a
    $\left(\frac{\|\Xi + \epsilon_{\text{ori}}\Lambda\|}{1-\delta/2}, m+1, \epsilon_a\|\Xi + \epsilon_{\text{ori}}\Lambda\| + \epsilon_{\text{ori}}\right)$-block-encoding of $\Xi$.
    The block-encoding factor can be upper bounded as 
    \begin{equation}
        \frac{\|\Xi + \epsilon_{\text{ori}}\Lambda\|}{1-\delta/2} \leq \frac{\|\Xi\|( 1 + \epsilon_{\text{ori}} / \|\Xi\| ) }{1-\delta/2} \leq \frac{\|\Xi\| }{1-\delta}.  
    \end{equation}
\end{proof}

Theorem~\ref{thm:uniform-amplitude-amplification} indicates that we can boost the block-encoding factor of one-step propagator to 
\begin{equation}
    \alpha_\ell \leq \frac{\left\|\mathcal{T}e^{\left(\int_{\ell h}^{(\ell+1)h}A(s)\mathrm{d}s\right)}\right\|}{1-\delta}
\end{equation}
by paying an extra cost of $\mathcal{O}\left(\frac{1}{\delta \left\|\mathcal{T}e^{\left(\int_{\ell h}^{(\ell+1)h}A(s)\mathrm{d}s\right)}\right\|}
\log\left(\frac{1}{\epsilon_a \left\|\mathcal{T}e^{\left(\int_{\ell h}^{(\ell+1)h}A(s)\mathrm{d}s\right)}\right\|}\right)\right)$.

\begin{theorem}
    Given the block-encoding $P_\ell$ defined in Section~\ref{section:input-models} as a one-step propagator,
    we are able to construct an $(\widetilde{\alpha_\ell}, a_1 + 1, \|\mathcal{T}e^{\int_{\ell h}^{(\ell+1)h}A(s)\mathrm{d}s}\|\epsilon_a + \epsilon_0 + \epsilon_a\epsilon_0)$-block-encoding $\widetilde{P_\ell}$, satisfying
    \begin{equation}
        \prod_{\ell=0}^{M-1}\widetilde{\alpha_\ell} =
        \mathcal{O}\left(\prod_{\ell=0}^{M-1}\left\|\mathcal{T}e^{\int_{\ell h}^{(\ell+1)h}A(s)\mathrm{d}s}\right\|\right),
    \end{equation}
    using $\mathcal{O}\left(\alpha_A T\log(1/\epsilon_a)\right)$ applications of (controlled)-$P_\ell$ and its inverse.
\end{theorem}
\begin{proof}
Now we need to determine $\delta$ in Theorem~\ref{thm:uniform-amplitude-amplification}. Notice that
\begin{equation}
\label{eqn:delta-choice}
    (1-\delta)^{M} \prod_{\ell=0}^{M-1}\widetilde{\alpha_\ell}
    =
    \mathcal{O}\left(\prod_{\ell=0}^{M-1} \left\|\mathcal{T}e^{\left(\int_{\ell h}^{(\ell+1)h}A(s)\mathrm{d}s\right)}\right\|\right),
\end{equation}
what we want is $(1-\delta)^M = \Omega(1)$. Then we can choose $\delta = \frac{1}{M} = \frac{1}{T/h} = \frac{1}{\alpha_A T}$.
\end{proof}

Denote $\left\|\mathcal{T}e^{\left(\int_{\ell h}^{(\ell+1)h}A(s)\mathrm{d}s\right)}\right\|\epsilon_a + \epsilon_0 + \epsilon_a\epsilon_0=:\left\|\mathcal{T}e^{\left(\int_{\ell h}^{(\ell+1)h}A(s)\mathrm{d}s\right)}\right\|\epsilon_{\text{all}}$,
then the sequential application of $P_\ell$'s can now be turned into applying $\{\widetilde{P_\ell}\}$:
\begin{equation}
\begin{split}
    \left(\prod_{\ell=0}^{M-1}\widetilde{P_\ell}\right)
    \left(\ket{0^{a_1}}\otimes \ket{u(0)}\right)
    &= 
    \ket{0^{a_1}} \prod_{\ell = 0}^{M-1}\frac{\mathcal{T}e^{\int_{\ell h}^{(\ell+1) h}A(s)\mathrm{d}s} + \epsilon_\text{all}\|\mathcal{T}e^{\int_{\ell h}^{(\ell+1) h}A(s)\mathrm{d}s}\| \Lambda_\ell}{\widetilde{\alpha_\ell}}\ket{u(0)} + \ket{\perp}\\
    &=
    \ket{0^{a_1}} \frac{\mathcal{T}e^{\int_{0}^{T}A(s)\mathrm{d}s} + \Gamma}{\prod_{\ell = 0}^{M-1}\widetilde{\alpha_\ell}}\ket{u(0)} + \ket{\perp}
\end{split}
\end{equation}
to get the desired final state. 
Here, $\Gamma$ satisfies
\begin{equation}
\label{eqn:gamma-estimation}
    \|\Gamma\| \leq\left(\prod_{\ell=0}^{M-1}\left\|\mathcal{T}e^{\left(\int_{\ell h}^{(\ell+1)h}A(s)\mathrm{d}s\right)}\right\|\right) \sum_{\ell=1}^{M}
    \begin{pmatrix}M\\\ell\end{pmatrix}
    \epsilon_{\text{all}}^\ell 
    \leq  \left(\prod_{\ell=0}^{M-1}\left\|\mathcal{T}e^{\left(\int_{\ell h}^{(\ell+1)h}A(s)\mathrm{d}s\right)}\right\|\right)\bigg((1+\epsilon_{\text{all}})^M - 1\bigg).
\end{equation}

In order to determine the schedule of $\epsilon$ and $\epsilon_0$ for a given tolerance $\epsilon_{\text{tol}}$, we need the following lemma.
\begin{lemma}
\label{lem:error-ana}
    For two non-zero vectors $x$ and $\widetilde{x}$, we have
    \begin{equation}
        \|\ket{x} - \ket{\widetilde{x}}\| \leq \frac{2\|x - \widetilde{x}\|}{\|x\|}.
    \end{equation}
\end{lemma}
\begin{proof}
    \begin{equation}
    \begin{split}
        \|\ket{x} - \ket{\widetilde{x}}\| &=     
        \left\|\frac{x}{\|x\|} - \frac{\widetilde{x}}{\|\widetilde{x}\|}\right\|
        =
        \left\|\frac{x}{\|x\|} - \frac{\widetilde{x}}{\|x\|}
        +
        \frac{\widetilde{x}}{\|x\|} - \frac{\widetilde{x}}{\|\widetilde{x}\|}
        \right\| 
        \leq
        \frac{1}{\|x\|}\|x - \widetilde{x}\| + \|\widetilde{x}\|\left\|\frac{1}{\|x\|} - \frac{1}{\|\widetilde{x}\|}\right\|\\
        &= \frac{\|x - \widetilde{x}\|}{\|x\|} + \frac{\left|\|x\| - \|\widetilde{x}\|\right|}{\|x\|} \leq \frac{2 \|x - \widetilde{x}\|}{\|x\|}.
    \end{split}
    \end{equation}
\end{proof}
Based on this, the two vectors are $\frac{\mathcal{T}e^{\int_{0}^{T}A(s)\mathrm{d}s} + \Gamma}{\prod_{\ell = 0}^{M-1}\widetilde{\alpha_\ell}}\ket{u(0)}$ and $\frac{\mathcal{T}e^{\int_{0}^{T}A(s)\mathrm{d}s}}{\prod_{\ell = 0}^{M-1}\widetilde{\alpha_\ell}}\ket{u(0)} = \frac{u(T)}{\prod_{\ell=0}^{M-1}\widetilde{\alpha_\ell}\|u(0)\|}$. We require
\begin{equation}
    2\frac{\prod_{\ell=0}^{M-1}\widetilde{\alpha_\ell}\|u(0)\|}{\|u(T)\|}
    \frac{\|\Gamma\|}{\prod_{\ell=0}^{M-1}\widetilde{\alpha_\ell}}
    = \mathcal{O}\left(\alpha_A T \epsilon_{\text{all}} \frac{\prod_{\ell=0}^{M-1}\left\|\mathcal{T}e^{\left(\int_{\ell h}^{(\ell+1)h}A(s)\mathrm{d}s\right)}\right\|\|u(0)\|}{\|u(T)\|}\right) \leq \epsilon_{\text{tol}}.
\end{equation}
Define
\begin{equation}
    Q := \frac{\prod_{\ell=0}^{M-1}\left\|\mathcal{T}e^{\left(\int_{\ell h}^{(\ell+1)h}A(s)\mathrm{d}s\right)}\right\|\|u(0)\|}{\|u(T)\|},
\end{equation}
then we should set
\begin{equation}
    \epsilon_{\text{all}} = \mathcal{O}\left(\frac{\epsilon_{\text{tol}}}{\alpha_A T Q}\right),
    \quad
    \epsilon_0 = \mathcal{O}\left(\frac{\epsilon_{\text{tol}}}{\alpha_A T Q}\right),
    \quad
    \epsilon_a = \mathcal{O}\left(\frac{\epsilon_{\text{tol}}}{\alpha_A T Q}\right).
\end{equation}

The final step is about determining the number of rounds for amplitude amplification, which is simple
since
\begin{equation}
    \left\|
    \frac{\mathcal{T}e^{\int_{0}^{T}A(s)\mathrm{d}s} + \Gamma}{\prod_{\ell = 0}^{M-1}\widetilde{\alpha_\ell}}\ket{u(0)}
    \right\|
    =
    \mathcal{O}
    \left(
    \left\|
    \frac{\mathcal{T}e^{\int_{0}^{T}A(s)\mathrm{d}s}}{\prod_{\ell = 0}^{M-1}\widetilde{\alpha_\ell}}\ket{u(0)}
    \right\|
    \right) \sim 1/Q.
\end{equation}

Now, we count the overall complexity of the time-marching method. 
The number of queries to $O_u$ is simply the number of rounds of amplitude amplification, which is $\mathcal{O}(Q)$. 
To construct a block-encoding of one-step propagator with proper block-encoding factor, we need $\mathcal{O}(\alpha_A T \log(1/\epsilon_a))$ applications of $P_\ell$. 
There are $M = \mathcal{O}(\alpha_A T)$  of them. 
Meanwhile, each $P_\ell$ requires $\mathcal{O}\left( \frac{\log(1/\epsilon_0)}{\log\log(1/\epsilon_0)} \right)$ queries to $\text{HAM-T}_{A,q}$, so the number of queries to $\text{HAM-T}_{A,q}$ in a single round becomes 
\begin{equation}
    \mathcal{O}\left( \alpha_A^2 T^2   \log(1/\epsilon_a) \frac{\log(1/\epsilon_0)}{\log\log(1/\epsilon_0)} \right). 
\end{equation}
Replacing $\epsilon_{\text{tol}}$ by $\epsilon$, we can now conclude the final complexity is
\begin{equation}
\label{eqn:time-marching-final-homogeneous-complexity}
    \mathcal{O}
    \left(
    \alpha_A^2 Q T^2 \frac{\log^2(\alpha_A T Q/\epsilon)}{\log\log(\alpha_A T Q/\epsilon)}\;
    \text{Cost}(\text{HAM-T}_{A,q})
    +
    Q\; \text{Cost}(O_u)
    \right). 
\end{equation}

\subsubsection{Inhomogeneous case}
\label{section:time-marching-inhomo}
Although the inhomogeneous case was briefly mentioned in~\cite{fang2023time}, it has not been proposed in detail previously. 
Here, we present how to generalize the time-marching method for preparing final states of inhomogeneous ODEs and analyze its complexity. 

Let us first introduce the Duhamel's principle. Basically, for the inhomogeneous differential equation
\begin{equation}
\begin{split}
    \frac{\mathrm{d}u}{\mathrm{d}t} &= A(t) u(t) + b(t),\\
    u(0) &= u_0,
\end{split}
\end{equation}
the final solution $u(T)$ can be written as
\begin{equation}
    u(T) = \mathcal{T}e^{\int_0^T A(s') \mathrm{d}s'} u(0)
    +
    \int_0^T \mathcal{T}e^{\int_t^T A(s') \mathrm{d}s'}\; b(t) \mathrm{d}t.
\end{equation}
Suppose we are able to do $\mathcal{T}e^{\int_t^T A(s')\mathrm{d}s'}$, the above equation can be discretized as
\begin{equation}
\begin{split}
    u(T) &=
    \mathcal{T}e^{\int_0^T A(s') \mathrm{d}s'} u(0)
    + \sum_{j=0}^{M_I -1} \int_{jh}^{(j+1)h}\mathcal{T}e^{\int_t^T A(s') \mathrm{d}s'}\; b(t) \mathrm{d}t\\
    &\approx
    \mathcal{T}e^{\int_0^T A(s') \mathrm{d}s'} u(0)
    + \sum_{j=0}^{M_I - 1} \sum_{s=0}^{n-1} c_s \mathcal{T}e^{\int_{t_s + jh}^T A(s') \mathrm{d}s'}\; b(t_s+jh).
\end{split}
\end{equation}
Here $s$'s are indices of the discretization nodes $\{t_s\}$ on $[0,h]$, $h = T/M_I$,
$j$'s are indices of the sub-intervals of the same length.
$c_s$'s are coefficients depending on the numerical integration scheme.
It is worth noting that the $c_s$'s do not depend on $j$. 
The term $\mathcal{T}e^{\int_0^T A(s') \mathrm{d}s'} u(0)$ is the solution of the homogeneous ODE and has been efficiently implemented. 
Once we can implement $\sum_{j=0}^{M_I - 1} \sum_{s=0}^{n-1} c_s \mathcal{T}e^{\int_{t_s + jh}^T A(s') \mathrm{d}s'}\; b(t_s+jh)$ due to the inhomogeneous term, the final solution can be computed by another layer of LCU.

Now we discuss the implementation of the inhomogeneous term. 
Consider the select oracle 
\begin{equation}
\label{eqn:sel-oracle-time-marching}
    O_{\text{sel}} = 
    \sum_{s = 0}^{n-1}
    \sum_{j=0}^{M_I-1}
    \ket{s}\bra{s}\otimes \ket{j} \bra{j}
    \otimes \mathcal{T}e^{\int_{t_s + jh}^T A(s) \mathrm{d}s},
\end{equation}
where, with an abuse of notation, $\mathcal{T}e^{\int_{t_s + jh}^T A(s) \mathrm{d}s}$ should be understood as its block-encoding. 
An important observation is that 
\begin{equation}
\begin{split}
    \sum_{s=0}^{n-1}\sum_{j=0}^{M_I-1}\ket{s}\bra{s}\otimes \ket{j}\bra{j} \otimes\mathcal{T}e^{\int_{t_s + jh}^T A(s) \mathrm{d}s}
    &=
    \sum_{s=0}^{n-1}\ket{s}\bra{s} \otimes
    \prod_{j=0}^{M_I-1}\left(\left(\sum_{\ell=0}^j \ket{\ell}\bra{\ell} \right)\otimes \mathcal{T}e^{\int_{t_s+jh}^{t_{s}+(j+1)h} A(s) \mathrm{d}s}\right).
\end{split}
\end{equation}
\begin{remark}
Notice that 
\begin{equation}
\prod_{j=0}^{M_I-1}\left(\left(\sum_{\ell=0}^j \ket{\ell}\bra{\ell} \right)\otimes \mathcal{T}e^{\int_{t_s+jh}^{t_{s}+(j+1)h} A(s) \mathrm{d}s}\right)
\end{equation}
actually is equivalent to $\sum_{j}\ket{j}\bra{j} \otimes\mathcal{T}e^{\int_{t_s + jh}^{t_s + (M_I-1)h} A(s) \mathrm{d}s}$, which is slightly different to 
$\sum_{j}\ket{j}\bra{j} \otimes\mathcal{T}e^{\int_{t_s + jh}^T A(s) \mathrm{d}s}$.
However, they can be made the same by manually defining $A(t) = 0$ when $t \in (T, T+h]$
and do the integration on $[0, T+h]$.
\end{remark}

Define $b_{s,j} := \|b(t_s +jh)\|$,
we may construct the preparation oracles
\begin{equation}
\begin{split}
    O_{c,l} \ket{0} &= \frac{1}{\sqrt{\sum_{s,j} |c_sb_{s,j}\alpha_{s,j}^{\text{int}}}|}\sum_{s,j}\overline{\sqrt{c_s b_{s,j}\alpha_{s,j}^{\text{int}}}}\ket{s}\ket{j},\\
    O_{c,r} \ket{0} &= \frac{1}{\sqrt{\sum_{s,j} |c_sb_{s,j}\alpha_{s,j}^{\text{int}}|}}\sum_{s,j}\sqrt{c_sb_{s,j}\alpha_{s,j}^{\text{int}}}\ket{s}\ket{j},
\end{split}
\end{equation}
where $\alpha_{s,j}^{\text{int}}$ is the normalization factor of $\mathcal{T}e^{\int_{t_s+jh}^{T} A(s) \mathrm{d}s}$.
Set $d = \log_2(M_I) + \log_2(n)$,
we have
\begin{equation}
    \big((O_{c,l}^\dag \otimes I_a \otimes I) O_{\text{sel}} 
    O_{\text{init}}
    (O_{c,r}\otimes I_a \otimes I)\big)\ket{0^d}\ket{0^a}\ket{0}
    \approx
    \ket{0^{a+d}} \frac{\|u(T)\|\ket{u(T)}}{\sum_{s,j} |c_s b_{s,j}\alpha_{s,j}^{\text{int}}|} + \ket{\perp}.
\end{equation}

Let us firstly take a glance on the complexity of constructing the select oracle. To construct 
\begin{equation}
\left(\sum_{\ell=0}^j \ket{\ell}\bra{\ell} \right)\otimes \mathcal{T}e^{\int_{t_s+jh}^{t_{s}+(j+1)h} A(s) \mathrm{d}s},
\end{equation}
it only takes $\mathcal{O}(\alpha_A T\log(1/\epsilon_0)/\log\log(1/\epsilon_0))$ queries.
While there are $M_I = \mathcal{O}(\alpha_A T)$ steps, 
indicating 
\begin{equation}
\prod_{j=0}^{M_I-1}\left(\left(\sum_{\ell=0}^j \ket{\ell}\bra{\ell} \right)\otimes \mathcal{T}e^{\int_{t_s+jh}^{t_{s}+(j+1)h} A(s) \mathrm{d}s}\right)
\end{equation}
costs just as Equation~\eqref{eqn:time-marching-final-homogeneous-complexity}, but without the amplitude amplification rounds:
\begin{equation}
    \mathcal{O}
    \left(
    \alpha_A^2 T^2 \frac{\log(1/\epsilon_a)\log(1/\epsilon_{0})}{\log\log(1/\epsilon_{0})}\;
    \text{Cost}(\text{HAM-T}_{A,q})
    \right).
\end{equation}
We further need to do this for $n = \mathcal{O}(\log(\alpha_A T/\epsilon_{\text{int}}))$ times since the select oracle is also controlled on $\ket{s}\bra{s}$. Fortunately this can be done simultaneously using one additional compare operator. 
Roughly speaking the query complexity is still
\begin{equation}
\label{eqn:time-marching-inhomo-history-regular-select-complexity}
    \mathcal{O}
    \left(
    \alpha_A^2 T^2 \frac{\log(1/\epsilon_a)\log(1/\epsilon_{0})}{\log\log(1/\epsilon_{0})}\;
    \text{Cost}(\text{HAM-T}_{A,q})
    \right).
\end{equation}

Define 
\begin{equation}
    \left(\sum_{\ell=0}^{j}\ket{\ell}\bra{\ell}\right)\otimes \mathcal{T}e^{\left(\int_{t_s + jh}^{t_s + (j+1)h}A(s)\mathrm{d}s\right)}
    = \left(\sum_{\ell=0}^{j}\ket{\ell}\bra{\ell}\right)\otimes P_{s,j}
\end{equation}
and
$P_{s,j}$ is an $(\alpha_{s,j}, a_1, \left\|\mathcal{T}e^{\int_{t_s + jh}^{t_s + (j+1)h}A(s)\mathrm{d}s}\right\|\epsilon_{\text{all}})$-block-encoding. Here
$\alpha_{s,j} = \mathcal{O}\left(\left\|\mathcal{T}e^{\int_{t_s + jh}^{t_s + (j+1)h}A(s)\mathrm{d}s}\right\|\right)$, $\epsilon_1 = \left\|\mathcal{T}e^{\int_{t_s + jh}^{t_s + (j+1)h}A(s)\mathrm{d}s}\right\|\epsilon_{\text{all}} = \left\|\mathcal{T}e^{\int_{t_s + jh}^{t_s + (j+1)h}A(s)\mathrm{d}s}\right\|\epsilon_a + \epsilon_0 + \epsilon_a \epsilon_0$.

Now let us forward to the computation. For simplicity, the following computation did not explicitly write out $u(0)$, however one can just treat it as some special $b(t)$.
\begin{equation}
\begin{split}
&\left(
(O_{c,l}^\dag \otimes I_a \otimes I)  \left(\sum_{s=0}^{n-1}\ket{s}\bra{s}\otimes\prod_{j=0}^{M_I-1} P_{s,j}\right)
O_{\text{init}}
(O_{c,r}\otimes I_a \otimes I)\right)\ket{0^d}\ket{0^a}\ket{0}\\
&=
(O_{c,l}^\dag \otimes I_a \otimes I)  \left(\sum_{s=0}^{n-1}\ket{s}\bra{s}\otimes\prod_{j=0}^{M_I-1} P_{s,j}\right)
O_{\text{init}}\frac{1}{\sqrt{\sum_{s,j}|c_s b_{s,j}\alpha_{s,j}^{\text{int}}|}}\sum_{s,j} \sqrt{c_s b_{s,j}\alpha_{s,j}^{\text{int}}}\ket{s}\ket{j}\ket{0^a}\ket{0}\\
&=
(O_{c,l}^\dag \otimes I_a \otimes I)  \left(\sum_{s=0}^{n-1}\ket{s}\bra{s}\otimes\prod_{j=0}^{M_I-1} P_{s,j}\right)
\frac{1}{\sqrt{\sum_{s,j}|c_s b_{s,j}\alpha_{s,j}^{\text{int}}|}}\sum_{s,j} \sqrt{c_s b_{s,j}\alpha_{s,j}^{\text{int}}}\ket{s}\ket{j}\ket{0^a}\ket{b(t_s + jh)}\\
&=
(O_{c,l}^\dag \otimes I_a \otimes I)  
\frac{1}{\sqrt{\sum_{s,j}|c_s b_{s,j} \alpha_{s,j}^{\text{int}}|}}\sum_{s,j} \sqrt{c_s b_{s,j}\alpha_{s,j}^{\text{int}}}\ket{s}\ket{j}
\ket{0^a}\left(\prod_{k=j}^{M_I-1}\frac{\mathcal{T}e^{\int_{t_s + k h}^{t_s + (k+1)h} A(s) \mathrm{d}s} + \epsilon_1 \Lambda_{s,k}}{\alpha_{s,k}}\right)\ket{b(t_s+jh)} + \ket{\perp}\\
&=
\frac{1}{\sum_{s,j}|c_s b_{s,j}\alpha_{s,j}^{\text{int}}|}\sum_{s,j} \left(c_s b_{s,j}
\ket{0^{a+d}}\left(\mathcal{T}e^{\int_{t_s+jh}^{t_s + (M_I-1)h} A(s) \mathrm{d}s}\ket{b(t_s+jh)} + \Gamma_{s,j}\ket{b(t_s+jh)}\right) \right) + \ket{\perp}.\\
&=
\frac{1}{\|u_0\|\alpha + \sum_{s,j}|c_s b_{s,j}\alpha_{s,j}^{\text{int}}|}
\ket{0^{a+d}}
\bigg( \|u_0\|\left(\mathcal{T}e^{\int_{0}^{T} A(s) \mathrm{d}s}\ket{u(0)} + \Gamma_0 \right) + \\
&\qquad\qquad \sum_{s,j} c_s b_{s,j}\left(\mathcal{T}e^{\int_{t_s+jh}^{t_s + (M_I-1)h} A(s) \mathrm{d}s}\ket{b(t_s+jh)} + \Gamma_{s,j}\ket{b(t_s+jh)}\right) \bigg) + \ket{\perp}.
\end{split}
\end{equation}
Here $\|\Lambda_{s,k}\| \leq 1$, 
\begin{equation}
\begin{split}
\label{eqn:time-marching-gamma-inhomo-estimation}
\|\Gamma_{s,j}\| &\leq 
\left(\prod_{k=j}^{M_I-1}\left\|\mathcal{T}e^{\int_{t_s+kh}^{t_s+(k+1)h} A(s)\mathrm{d}s}\right\|\right)\sum_{\ell=1}^{M_I - j}\begin{pmatrix}M_I - j\\\ell\end{pmatrix}
\epsilon_{\text{all}}^\ell 
\lesssim M_I \epsilon_{\text{all}}\left(\prod_{k=j}^{M_I-1}\left\|\mathcal{T}e^{\int_{t_s+kh}^{t_s+(k+1)h} A(s)\mathrm{d}s}\right\|\right),\\
\|\Gamma_0\| &\leq 
\left(\prod_{k=0}^{M_I-1}\left\|\mathcal{T}e^{\int_{kh}^{(k+1)h} A(s)\mathrm{d}s}\right\|\right)\sum_{\ell=1}^{M_I}\begin{pmatrix}M_I\\\ell\end{pmatrix}
\epsilon_{\text{all}}^\ell 
\lesssim M_I \epsilon_{\text{all}}\left(\prod_{k=j}^{M_I-1}\left\|\mathcal{T}e^{\int_{kh}^{(k+1)h} A(s)\mathrm{d}s}\right\|\right),
\end{split}
\end{equation}
and $\alpha_{s,j}$ is chosen to make
\begin{equation}
\begin{split}
    \alpha_{s,j}^{\text{int}} &= \prod_{k = j}^{M_I-1} \alpha_{s,k}
    = \mathcal{O}\left(
    \prod_{k = j}^{M_I-1} \left\|\mathcal{T}e^{\int_{t_s + kh}^{t_s + (k+1)h}A(s)\mathrm{d}s}\right\|
    \right), \quad
    \alpha = 
    \mathcal{O}\left(
    \prod_{k = 0}^{M_I-1} \left\|\mathcal{T}e^{\int_{kh}^{ (k+1)h}A(s)\mathrm{d}s}\right\|\right)
\end{split}
\end{equation}
hold true.

Let us now derive an estimation on
\begin{equation}
\begin{split}
    \|u(0)\|\alpha + 
    \sum_{s,j} |c_s b_{s,j}\alpha_{s,j}^{\text{int}}|
    &=
    \mathcal{O}\left(
    \|u(0)\| \prod_{k=0}^{M_I-1}\left\|\mathcal{T}e^{\int_{kh}^{(k+1)h}A(\tau) \mathrm{d}\tau}\right\|+
    \sum_{s,j}|c_s b_{s,j}|\prod_{k=j}^{M_I-1}\left\|\mathcal{T}e^{\int_{t_s+kh}^{t_s+(k+1)h}A(\tau) \mathrm{d}\tau}\right\|\right)\\
    &= \mathcal{O}\left(Q_0 \left(\|u(0)\| + \sum_{s,j}|c_s b_{s,j}|\right) \right)
    =\mathcal{O}\left(Q_0(\|b\|_{L_1} + \|u(0)\|) \right).
\end{split}
\end{equation}
Here, $Q_0$ is a constant factor which depends only on $\|\mathcal{T}e^{\int_{t+kh}^{t+(k+1)h}A(s)\mathrm{d}s}\|$, for $t \in [0,T]$.

Next, we take care of the error analysis. Given a tolerance $\epsilon_{\text{tol}}$, we have two resources of the error:
the first is the numerical integration, and the second one comes from our imperfect block-encodings. 

Notice that 
\begin{equation}
    \left\|\|u(0)\| \mathcal{T}e^{\int_0^TA(s)\mathrm{d}s}\ket{u(0)} + \sum_{s,j} c_s b_{s,j} \mathcal{T}e^{\int_{t_s+jh}^{t_s + (M_I-1)h} A(s) \mathrm{d}s}\ket{b(t_s+jh)}\right\| = \mathcal{O}(\|u(T)\|),
\end{equation}
and define
\begin{equation}
    Q := \mathcal{O}\left(\frac{Q_0(\|b\|_{L_1} + \|u(0)\|)}{\|u(T)\|}\right).
\end{equation}

Leveraging Lemma~\ref{lem:error-ana}, we restrict
\begin{equation}
\label{eqn:time-marching-final-inhomo-error-restriction}
    \frac{1}{\alpha \|u(0)\| + \sum_{s,j} |c_s b_{s,j}\alpha_{s,j}^{\text{int}}|}\frac{\|u(0)\|\|\Gamma_0\| + \sum_{s,j} c_s b_{s,j}\|\Gamma_{s,j}\|}{1/Q} \leq \epsilon_{\text{tol}}/4.
\end{equation}
Using the estimate in Equation~\eqref{eqn:time-marching-gamma-inhomo-estimation},
the Equation~\eqref{eqn:time-marching-final-inhomo-error-restriction} can be bounded if
\begin{equation}
      M_I \epsilon_{\text{all}} = 
    \mathcal{O}(\epsilon_{\text{tol}}/Q).
\end{equation}
This means that we need to take
$\epsilon_{\text{all}} = \mathcal{O}\left(\frac{\epsilon_{\text{tol}}}{T\alpha_A Q}\right)$, meaning
$\epsilon_0 = \mathcal{O}\left(\frac{\epsilon_{\text{tol}}}{T\alpha_A Q}\right), \epsilon_{a} = \mathcal{O}\left(\frac{\epsilon_{\text{tol}}}{T\alpha_A Q}\right)$.
For the integration part, we can simply put $\epsilon_{\text{int}} = \mathcal{O}(\epsilon_{\text{tol}}/Q)$.

It is clear that the number of rounds required for applying the amplitude amplification is again defined as $Q$. Then the cost of final preparation in the inhomogeneous setting should be
\begin{equation}
\label{eqn:time-marching-final-state-inhomo-complexity}
    \mathcal{O}\left(\alpha_A^2 Q T^2 \frac{\log^2(\alpha_A T Q/\epsilon_{\text{tol}})}{\log\log(\alpha_A TQ/\epsilon_{\text{tol}})}\;\text{Cost}(\text{HAM-T}_{A,q}) + Q\;\text{Cost}(O_{\text{init}})\right).
\end{equation}

\subsubsection{Dissipative inhomogeneous case}\label{section:time-marching-inhomo-dissipative}

Now we forward to the dissipative case.
The first point we notice is that the choice of $\delta$ in Theorem~\ref{thm:uniform-amplitude-amplification} no longer depends on $T$.
Specifically, if the dynamics is dissipative, namely
\begin{equation}
    A(t) + A^\dag(t) \leq -2\eta < 0,
\end{equation}
indicating
\begin{equation}
    \left\|
    \mathcal{T}e^{\int_{t_0}^{t_1} A(s) \mathrm{d}s}
    \right\|
    \leq 
    e^{-\eta (t_1 - t_0)},
\end{equation}
then instead of achieving the same as that in Equation~\eqref{eqn:delta-choice}, we are choosing $\delta$ such that we can get an $\Omega(1)$ normalization factor in the end.
This is equivalent to hoping
\begin{equation}
    (1-\delta)^k = \Theta(e^{-k\eta h})
\end{equation}
holds for all $k$. Thus $1-\delta = \Theta(e^{-\eta h})$, i.e. $\delta = \Theta(\eta h) = \Theta(\eta/\alpha_A)$. 

Based on the discussion above, we are able to improve the scaling in $T$ from quadratic to linear due to the improvement in the uniform amplitude amplification step. 
The corresponding complexity in Equation~\eqref{eqn:time-marching-final-state-inhomo-complexity} then can be improved to
\begin{equation}
    \label{eqn:time-marching-final-state-inhomo-complexity-dissi-first-improvement}
        \mathcal{O}\left(\alpha_A^2 Q (T/\eta) \frac{\log^2(\alpha_A T Q/\epsilon_{\text{tol}})}{\log\log(\alpha_A TQ/\epsilon_{\text{tol}})}\;\text{Cost}(\text{HAM-T}_{A,q}) + Q\;\text{Cost}(O_{\text{init}})\right).
\end{equation}
Here $Q$ should be 
\begin{equation}
    Q = \mathcal{O}\left(\frac{\|b\|_{L_1} + \|u(0)\|}{\|u(T)\|}\right).
\end{equation}

To further remove the remaining linear dependence on $T$, our idea is to determine a time period with length $T_0$ such that
\begin{equation}\label{eqn:time_marching_truncation_error}
    \left\|
    \left|
    \int_{T - T_0}^T \mathcal{T}e^{\int_t^T A(s)\mathrm{d}s}b(t)\mathrm{d}t
    \right\rangle
    -
    \left|
    \mathcal{T}e^{\int_0^T A(s) \mathrm{d}s} u(0)
    +
    \int_0^T \mathcal{T}e^{\int_t^T A(s) \mathrm{d}s}\; b(t) \mathrm{d}t
    \right\rangle
    \right\| \leq \epsilon_{\text{truncate}}. 
\end{equation}
The choice of $T_0$ is given in the next result. 

\begin{lemma}\label{lem:T0_time_marching_final}
    Let $\epsilon_{\text{truncate}} > 0$, and $T$ is sufficiently large such that $T \geq \frac{1}{\eta} \log\left( \frac{4\|u(0)\|}{ \epsilon_{\text{truncate}} \|u(T)\| } \right)$. 
    Then, in order to garantee Equation~\eqref{eqn:time_marching_truncation_error}, it suffices to choose 
    \begin{equation}
        T_0 = \mathcal{O}\left(\frac{1}{\eta}\log\left(\frac{\max_t \|b(t)\|}{\eta\epsilon_{\text{truncate}}\|u(T)\|}\right)\right). 
    \end{equation}
\end{lemma}

\begin{proof}
    Use Lemma~\ref{lem:error-ana}, we can bound the left hand side of Equation~\eqref{eqn:time_marching_truncation_error} by 
\begin{equation}
\begin{split}
    & \quad \frac{2}{\|u(T)\|}\left\|\mathcal{T}e^{\int_0^T A(s)\mathrm{d}s}u(0) + \int_0^{T-T_0}\mathcal{T}e^{\int_t^T A(s)\mathrm{d}s}b(t)\mathrm{d}t\right\| \\
    &\leq
    \frac{2}{\|u(T)\|} \left( \left\|\mathcal{T}e^{\int_0^T A(s)\mathrm{d}s}\right\|\|u(0)\| + \max_t\|b(t)\|\int_0^{T-T_0}\left\|\mathcal{T}e^{\int_t^T A(s)\mathrm{d}s}\right\|\mathrm{d}t \right) \\
    &\leq
     \frac{2}{\|u(T)\|} \left( e^{-\eta T}\|u(0)\| + \max_t\|b(t)\|\int_0^{T-T_0}e^{-\eta (T - t)}\mathrm{d}t \right) \\
     & \leq \frac{2 \|u(0)\| }{ \|u(T)\| }e^{-\eta T} + \frac{2 \max_t\|b(t)\|}{\|u(T)\|} \frac{e^{-\eta T_0} }{\eta}. 
\end{split}
\end{equation}
By the assumption of $T$, the first term $ \frac{2 \|u(0)\| }{ \|u(T)\| }e^{-\eta T}$ is already upper bounded by $\epsilon_{\text{truncate}}/2$, so we only need to choose $T_0$ such that 
\begin{equation}
    \frac{2 \max_t\|b(t)\|}{\|u(T)\|} \frac{e^{-\eta T_0} }{\eta} \leq \frac{\epsilon_{\text{truncate}}}{2}, 
\end{equation}
which is true if
\begin{equation}
    T_0 = \mathcal{O}\left( \frac{1}{\eta} \log\left(\frac{\max_t \|b(t)\|}{\eta\epsilon_{\text{truncate}}\|u(T)\|}\right)\right).
\end{equation}
\end{proof}

Lemma~\ref{lem:T0_time_marching_final} indicates that the effective simulation time of dissipative ODEs is $T_0 = \mathcal{O}\left(\frac{1}{\eta}\log\left(\frac{\max_t \|b(t)\|}{\eta\epsilon_{\text{truncate}}\|u(T)\|}\right)\right)$, which depends on the rate of dissipation but becomes independent of $T$ as long as it is sufficiently large. 
If $T$ does not satisfy the condition in Lemma~\ref{lem:T0_time_marching_final}, then we can simply simulate the entire dynamics. 
As a result, in any case, the overall simulation time is always bounded by $\mathcal{O}\left(\frac{1}{\eta}\log\left(\frac{\|u_0\|+\max_t \|b(t)\|}{\eta\epsilon_{\text{truncate}}\|u(T)\|}\right)\right)$. 

To bound the overall error in the quantum state by $\epsilon_{\text{tol}}$, we need to set $\epsilon_{\text{truncate}} = \mathcal{O}(\epsilon_{\text{tol}}/Q)$. 
Then the overall complexity for final state preparation in the dissipative case is
\begin{equation}
    \label{eqn:time-marching-final-state-inhomo-complexity-dissi-last-improvement}
    \begin{split}
        &\mathcal{O}\left(\alpha_A^2 Q \frac{\log\left(\frac{Q(\|u_0\|+\max_t \|b(t)\|)}{\eta\epsilon_{\text{tol}}\|u(T)\|}\right)}{\eta^2} \frac{\log^2\left(\frac{\alpha_A Q \log\left(\frac{Q ( \|u_0\| + \max_t \|b(t)\|)}{\eta\epsilon_{\text{tol}}\|u(T)\|}\right)}{\eta\epsilon_{\text{tol}}}\right)}{\log\log\left(\frac{\alpha_A Q \log\left(\frac{Q ( \|u_0\| + \max_t \|b(t)\|) }{\eta\epsilon_{\text{tol}}\|u(T)\|}\right)}{\eta\epsilon_{\text{tol}}}\right)}\;\text{Cost}(\text{HAM-T}) + Q\;\text{Cost}(O_{\text{init}})\right)\\
        &=\widetilde{\mathcal{O}}\left(\frac{\alpha_A^2 Q}{\eta^2} \frac{\log^3\left(1/\epsilon_{\text{tol}}\right)}{\log\log\left(1/\epsilon_{\text{tol}}\right)}\;\text{Cost}(\text{HAM-T}_{A,q}) + Q\;\text{Cost}(O_{\text{init}})\right).
    \end{split}
\end{equation}

\subsection{History state preparation}
\subsubsection{Homogeneous case}
The goal of history state preparation is to prepare the state proportional to 
\begin{equation}
    \sum_{k=0}^{M-1}\ket{k}\otimes u(kh), 
\end{equation}
where $u(kh) = \mathcal{T}e^{\int_0^{kh}A(s)\mathrm{d}s} u(0)$, $M = T/h = \mathcal{O}(\alpha_A T)$. 

Consider
\begin{equation}
    (H^{\otimes m} \otimes O_u)\ket{0^{m+n}} = \frac{1}{\sqrt{M}}\sum_{k=0}^{M-1}\ket{k}\ket{u(0)},
\end{equation}
we then apply controlled $\widetilde{P_\ell}$'s sequentially:
\begin{equation}
\label{eqn:time-marching-history-state}
\begin{split}
    &\left(\prod_{\ell=0}^{M-1} \left(I - \sum_{k=0}^{\ell} \ket{\ell}\bra{\ell}\right)\otimes \widetilde{P_\ell}\right) 
    \left(\ket{0^{a_1}}\otimes\left(\frac{1}{\sqrt{M}}\sum_{k=0}^{M-1}\ket{k}\ket{u(0)}\right)\right)\\
    &=
    \frac{1}{\sqrt{M}}\ket{0^{a_1}}\sum_{k=0}^{M-1}\ket{k}\prod_{\ell=0}^{k-1}\frac{\left(\mathcal{T}e^{\left(\int_{\ell h}^{(\ell+1)h}A(s)\mathrm{d}s\right)}+\epsilon_\text{all}\|\mathcal{T}e^{\int_{\ell h}^{(\ell+1) h}A(s)\mathrm{d}s}\|\Lambda_\ell\right)}{\widetilde{\alpha_\ell}}\ket{u(0)}
    +\ket{\perp}\\
    &=\frac{1}{\sqrt{M}}\ket{0^{a_1}}\sum_{k=0}^{M-1}\ket{k}\frac{\mathcal{T}e^{\left(\int_{0}^{kh}A(s)\mathrm{d}s\right)} + \Gamma_k}{\prod_{\ell=0}^k \widetilde{\alpha_\ell}}\ket{u(0)} + \ket{\perp}.
\end{split}
\end{equation}
Here the norm of $\Gamma_k$ can be bounded just as that in Equation~\eqref{eqn:gamma-estimation}:
\begin{equation}
    \|\Gamma_k\| \leq \left(\prod_{\ell=0}^{k-1}\left\|\mathcal{T}e^{\left(\int_{\ell h}^{(\ell+1)h}A(s)\mathrm{d}s\right)}\right\|\right)\bigg((1+\epsilon_{\text{all}})^k - 1\bigg).
\end{equation}

Note that the normalization factor $Q$ can be expressed as
\begin{equation}
\begin{split}
    \left(\frac{1}{M}\sum_{k=0}^{M-1}\frac{\|u(kh)\|^2}{\|u(0)\|^2 \prod_{\ell=0}^{k-1}\widetilde{\alpha_\ell}^2}\right)^{1/2}
    &\sim
    \left(\frac{1}{M}\sum_{k=0}^{M-1}\frac{\left\|u(kh)\right\|^2}{ 
    \left(\prod_{\ell=0}^{k-1} \left\|\mathcal{T}e^{\left(\int_{\ell h}^{(\ell+1)h}A(s)\mathrm{d}s\right)}\right\|^2\right)\|u(0)\|^2}\right)^{1/2}\\
    &=: 1/Q.
\end{split}
\end{equation}

Use Lemma~\ref{lem:error-ana}, we want
\begin{equation}
    2Q\left(\frac{1}{M}\sum_{k=0}^{M-1}\left\|\frac{\Gamma_k \ket{u(0)}}{\prod_{\ell=0}^{k-1}\widetilde{\alpha_\ell}}\right\|^2\right)^{1/2} \leq \epsilon_{\text{tol}}.
\end{equation}
This inspires us to choose $\epsilon_{\text{all}} = \mathcal{O}\left(\frac{\epsilon_{\text{tol}}}{\alpha_A T Q}\right)$, which further indicates that
\begin{equation}
    \epsilon_{0} = \mathcal{O}\left(\frac{\epsilon_{\text{tol}}}{\alpha_A T Q}\right),\quad
    \epsilon_{a} = \mathcal{O}\left(\frac{\epsilon_{\text{tol}}}{\alpha_A T Q}\right).
\end{equation}
Notice that for the two operators 
\begin{equation}
    \prod_{\ell=0}^{M-1} \left(I - \sum_{k=0}^{\ell} \ket{\ell}\bra{\ell}\right)\otimes \widetilde{P_\ell}
    \quad \text{and} \quad
   \prod_{\ell=0}^{M-1}  \widetilde{P_\ell},
\end{equation}
the queries for $\text{HAM-T}_{A,q}$ we need is the same. So the complexity here scales the same as that of homogeneous final state preparation in Equation~\eqref{eqn:time-marching-final-homogeneous-complexity}:
\begin{equation}
\label{eqn:time-marching-history-homogeneous-complexity}
    \mathcal{O}
    \left(
    \alpha_A^2 Q T^2 \frac{\log^2(\alpha_A T Q/\epsilon)}{\log\log(\alpha_A T Q/\epsilon)}\;
    \text{Cost}(\text{HAM-T}_{A,q})
    +
    Q\; \text{Cost}(O_u)
    \right). 
\end{equation}

\subsubsection{Dissipative homogeneous case}

For the dissipation case, using a similar argument in Equation~\eqref{eqn:time-marching-final-state-inhomo-complexity-dissi-first-improvement}, the complexity above can be improved into
\begin{equation}
\label{eqn:time-marching-history-state-homo-complexity-dissi-first-improvement}
    \mathcal{O}\left(\alpha_A^2 Q (T/\eta) \frac{\log^2(\alpha_A T Q/\epsilon_{\text{tol}})}{\log\log(\alpha_A TQ/\epsilon_{\text{tol}})}\;\text{Cost}(\text{HAM-T}_{A,q}) + Q\;\text{Cost}(O_{\text{init}})\right),
\end{equation}
however the $Q$ is redefinied as
\begin{equation}
    1/Q :=\left(\frac{1}{M}\sum_{k=0}^{M-1}\frac{\|u(kh)\|^2}{\|u(0)\|^2}\right)^{1/2}.
\end{equation}

The next step is to determine the simulation time $T_0 = M_0 h$, where
\begin{equation}
    \left\|
    \left|
        \sum_{k=0}^{M_0-1} \ket{k}\bra{k}\otimes \ket{u(kh)}
    \right\rangle
    -
    \left|
        \sum_{k=0}^{M-1} \ket{k}\bra{k}\otimes \ket{u(kh)}
    \right\rangle
    \right\|
    \leq \epsilon_{\text{truncate}}.
\end{equation}
Lemma~\ref{lem:error-ana} tells us we can turn this into
\begin{equation}
    2\left\|
        \sum_{k=M_0}^{M} \ket{k}\bra{k}\otimes \ket{u(kh)}
    \right\|
    =
    2\left(\sum_{k=M_0}^{M-1}\|u(kh)\|^2\right)^{1/2}
    \leq \epsilon_{\text{truncate}}\left(\sum_{k=0}^{M_0-1}\|u(kh)\|^2\right)^{1/2}.
\end{equation}
Now we need the Lemma~\ref{lemma:dissi-bound} to lower and also upper bound $\sum_k \|u(kh)\|^2$. 
We first use Equation~\eqref{eqn:dissi-upper-bound} to upper bound
\begin{equation}
    \left(\sum_{k=M_0}^{M-1}\|u(kh)\|^2\right)^{1/2}
    \leq
    \left(\sum_{k=M_0}^{M-1}e^{-2\eta kh}\|u(0)\|^2\right)^{1/2}
    =
    \left(\frac{e^{-2\eta M_0 h} - e^{-2\eta M h}}{1 - e^{-2\eta h}}\right)^{1/2} \|u(0)\|.
\end{equation}
Then we use~\cite[Lemma 5]{an2024fast} to lower bound
\begin{equation}
    \left(\sum_{k=0}^{M_0-1}\|u(kh)\|^2\right)^{1/2}
    \geq
    \left(\sum_{k=0}^{M_0-1}e^{-2\alpha_A kh}\|u(0)\|^2\right)^{1/2}
    =
    \left(\frac{1 - e^{-2\alpha_A M_0 h}}{1 - e^{-2\alpha_A h}}\right)^{1/2} \|u(0)\|.
\end{equation}
It is then sufficient to choose
\begin{equation}
    \label{eqn:homo-history-dissi-time-choice}
    M_0 = \mathcal{O}\left(\frac{\alpha_A}{\eta}\log( 1/\epsilon_{\text{truncate}})\right),\qquad
    T_0 = M_0 h = \mathcal{O}\left(\frac{1}{\eta}\log(1/\epsilon_{\text{truncate}})\right).
\end{equation}
It is clear that $\epsilon_{\text{truncate}}$ should scles as $\mathcal{O}\left(\epsilon_{\text{tol}}/Q\right)$,
then the complexity in Equation~\eqref{eqn:time-marching-history-state-homo-complexity-dissi-first-improvement} now is turned into
\begin{equation}
\label{eqn:time-marching-history-state-homo-complexity-dissi-last-improvement}
\begin{split}
    &\quad \mathcal{O}\left(\alpha_A^2 (Q/\eta^2) \frac{\log( Q/\epsilon_{\text{tol}})\log^2(\frac{\alpha_A Q\log(Q/\epsilon_{\text{tol}})}{\eta\epsilon_{\text{tol}}})}{\log\log(\frac{\alpha_A Q\log(Q/\epsilon_{\text{tol}})}{\eta\epsilon_{\text{tol}}})}\;\text{Cost}(\text{HAM-T}_{A,q}) + Q\;\text{Cost}(O_{\text{init}})\right) \\ 
    & = \widetilde{\mathcal{O}} \left(\frac{\alpha_A^2 Q}{\eta^2} \frac{\log^3(1/\epsilon_{\text{tol}})}{\log\log(1/\epsilon_{\text{tol}})}\;\text{Cost}(\text{HAM-T}_{A,q}) + Q\;\text{Cost}(O_{\text{init}})\right). 
\end{split}
\end{equation}

\subsubsection{Inhomogeneous case}
\label{section:time-marching-his-inhomo}
Now we consider preparing a history state in the inhomogeneous case.
Here we need the select oracle defined as follows:
\begin{equation}
\begin{split}
\label{eqn:time-marching-inhomo-history-sel}
    O_{\text{sel}} &= 
    \sum_{s = 0}^{n-1}
    \sum_{j=0}^{M_I-1}
    \sum_{k=0}^{j-1}
    \ket{s}\bra{s}\otimes \ket{j} \bra{j} \otimes \ket{k}\bra{k}
    \otimes \mathcal{T}e^{\int_{t_s + kh}^{t_s +jh} A(s) \mathrm{d}s}\\
    &=
    \sum_{j=0}^{M_I-1}
    \ket{j} \bra{j}
    \otimes
    \left(
    \sum_{s = 0}^{n-1}\sum_{k=0}^{j-1}
    \ket{s}\bra{s}\otimes \ket{k}\bra{k} \otimes
    \mathcal{T}e^{\int_{t_s + kh}^{t_s + jh} A(s) \mathrm{d}s}\right)\\
    &=
    \sum_{j=0}^{M_I-1}
    \ket{j} \bra{j}
    \otimes
    \left(
    \sum_{s = 0}^{n-1}
    \ket{s}\bra{s}\otimes \prod_{k=0}^{j-1}\left(\sum_{\ell=0}^k\ket{\ell}\bra{\ell}\right)\otimes
    \mathcal{T}e^{\int_{t_s + k h}^{t_s + (k+1)h} A(s) \mathrm{d}s}\right)\\
    &=
    \sum_{s = 0}^{n-1}
    \ket{s}\bra{s}\otimes 
    \sum_{j=0}^{M_I-1}
    \ket{j} \bra{j}\otimes
    \prod_{k=0}^{j-1}\left(\sum_{\ell=0}^k\ket{\ell}\bra{\ell}\right)
    \otimes
    \mathcal{T}e^{\int_{t_s + k h}^{t_s + (k + 1)h} A(s) \mathrm{d}s}\\
    &=
    \sum_{s = 0}^{n-1}
    \ket{s}\bra{s}\otimes 
    \sum_{j=0}^{M_I-1}
    \ket{j} \bra{j}\otimes
    \prod_{k=0}^{M_I-1}\mathbbm{1}_{k < j}\left(\sum_{\ell=0}^k \ket{\ell}\bra{\ell}\right)
    \otimes
    \mathcal{T}e^{\int_{t_s + k h}^{t_s + (k + 1)h} A(s) \mathrm{d}s}\\
    &=
    \sum_{s = 0}^{n-1}
    \sum_{j=0}^{M_I-1}
    \ket{s}\bra{s}\otimes 
    \ket{j} \bra{j}\otimes
    \prod_{k=0}^{M_I-1}\mathbbm{1}_{k < j}E_{s,k}.
\end{split}
\end{equation}
Here $E_{s,k} = \left(\sum_{\ell=0}^k \ket{\ell}\bra{\ell}\right)
\otimes
\mathcal{T}e^{\int_{t_s + k h}^{t_s + (k + 1)h} A(s) \mathrm{d}s}$.
\begin{remark}
    Just like when we are constructing the final state, the precise select oracle should be
    \begin{equation}
    \sum_{s = 0}^{n-1}
    \sum_{j=0}^{M_I-1}
    \sum_{k=0}^{j-1}
    \ket{s}\bra{s}\otimes \ket{j} \bra{j} \otimes \ket{k}\bra{k}
    \otimes \mathcal{T}e^{\int_{t_s + kh}^{jh} A(s) \mathrm{d}s}.
    \end{equation}
    But it can be solved by additionally define some intervals of length $h$ where $A(t)$ is defined to be an all zero matrix.
\end{remark}

Since $\mathbbm{1}_{k<j}E_{s,k}$ actually does not depend on index $j$, the construction in Equation~\eqref{eqn:time-marching-inhomo-history-sel} requires
the same queries as that in Equation~\eqref{eqn:time-marching-inhomo-history-regular-select-complexity}:
\begin{equation}
    \mathcal{O}\left(
    \alpha_A^2 T^2 \frac{\log(1/\epsilon_a)\log(1/\epsilon_0)}{\log\log(1/\epsilon_0)}
    \right).
\end{equation}
And $\mathcal{T}e^{\int_{t_s+kh}^{t_s + (k+1)h}A(s)\mathrm{d}s}$ are now actually
$(\alpha_{s,k}, a_1, \left\|\mathcal{T}e^{\int_{t_s + kh}^{t_s + (k+1)h}A(s)\mathrm{d}s}\right\|\epsilon_{\text{all}})$-block-encodings $P_{s,k}$. Here
$\alpha_{s,k} = \mathcal{O}\left(\left\|\mathcal{T}e^{\int_{t_s + kh}^{t_s + (k+1)h}A(s)\mathrm{d}s}\right\|\right)$, $\epsilon_1 = \left\|\mathcal{T}e^{\int_{t_s + kh}^{t_s + (k+1)h}A(s)\mathrm{d}s}\right\|\epsilon_{\text{all}} = \left\|\mathcal{T}e^{\int_{t_s + kh}^{t_s + (k+1)h}A(s)\mathrm{d}s}\right\|\epsilon_a + \epsilon_0 + \epsilon_a \epsilon_0$.
For simplicity we also denote $P_{s,k}$ controlled on $\left(\sum_{\ell=0}^k \ket{\ell}\bra{\ell}\right)$ as $P_{s,k}$.

It is clear that if we sum over indices $s$ and $k$, we are able to get the history state. In order to achieve this, we construct the state preparation oracles
\begin{equation}
\label{eqn:his-pre-oracle}
\begin{split}
    O_{c,l}^{\text{history}}\left(\frac{1}{\sqrt{M_I}}\sum_j \ket{j}\right) \ket{0}\ket{0} &= \frac{1}{\sqrt{M_I}}\sum_j \ket{j}\frac{1}{\sqrt{\sum_{s,k} |c_sb_{s, k}\alpha_{s,k,j}^{\text{int}}}|}\sum_{s,k}\overline{\sqrt{c_s b_{s,k}\alpha_{s, k,j}^{\text{int}}}}\ket{s}\ket{k},\\
    O_{c,r}^{\text{history}} \left(\frac{1}{\sqrt{M_I}}\sum_j \ket{j}\right)\ket{0}\ket{0} &= \frac{1}{\sqrt{M_I}}\sum_j \ket{j}\frac{1}{\sqrt{\sum_{s,k} |c_sb_{s,k}\alpha_{s,k,j}^{\text{int}}|}}\sum_{s,k}\sqrt{c_sb_{s,k}\alpha_{s,k,j}^{\text{int}}}\ket{s}\ket{k}.
\end{split}
\end{equation}
Here $b_{s,k} = \|b(t_s + kh)\|$, $\alpha_{s,k,j}^{\text{int}}$ is the normalization facor of our construction to $\mathcal{T}e^{\int_{t_s+kh}^{t_s + jh}}$.

Set $d_1 = \log_2(n), d_2 = \log_2(M_I)$,
with some excahnge of registers, we do the following computation:
{\footnotesize
\begin{equation}
\begin{split}
&
\left(
((O_{c,l}^{\text{history}})^\dag \otimes I_a \otimes I) 
\left(\sum_{j=0}^{M_I-1}\sum_{s=0}^{n-1}\ket{sj}\bra{sj}\otimes\prod_{k=0}^{M_I-1}\mathbbm{1}_{k<j} P_{s,k}\right)
O_{\text{init}} 
(O_{c,r}^{\text{history}}\otimes I_a \otimes I)\right)\left(\sum_j \frac{\ket{j}}{\sqrt{M_I}}\right)\ket{0^{d_1}}\ket{0^{d_2}}\ket{0^a}\ket{0}\\
&=
\left(
((O_{c,l}^{\text{history}})^\dag \otimes I_a \otimes I) 
\left(\sum_{j=0}^{M_I-1}\sum_{s=0}^{n-1}\ket{sj}\bra{sj}\otimes\prod_{k=0}^{M_I-1}\mathbbm{1}_{k<j} P_{s,k}\right)
O_{\text{init}}
\right)\sum_j \frac{\ket{j}}{\sqrt{M_I}}
\frac{\sum_{s,k}\sqrt{c_s b_{s,k}\alpha_{s, j,k}^{\text{int}}}\ket{s}\ket{k}}{\sqrt{\sum_{s,k} |c_sb_{s, k}\alpha_{s,j,k}^{\text{int}}}|}
\ket{0^a}\ket{0}\\
&=
\left(
((O_{c,l}^{\text{history}})^\dag \otimes I_a \otimes I) 
\left(\sum_{j=0}^{M_I-1}\sum_{s=0}^{n-1}\ket{sj}\bra{sj}\otimes\prod_{k=0}^{M_I-1}\mathbbm{1}_{k<j} P_{s,k}\right)
\right)\sum_j \frac{\ket{j}}{\sqrt{M_I}}
\frac{\sum_{s,k}\sqrt{c_s b_{s,k}\alpha_{s, j,k}^{\text{int}}}\ket{s}\ket{k}}{\sqrt{\sum_{s,k} |c_sb_{s, k}\alpha_{s,j,k}^{\text{int}}}|}
\ket{0^a}\ket{b(t_s + kh)}\\
&=\left(
((O_{c,l}^{\text{history}})^\dag \otimes I_a \otimes I)  
\right)\sum_j  \frac{\ket{j}}{\sqrt{M_I}}
\frac{\sum_{s,k}\sqrt{c_s b_{s,k}\alpha_{s, j, k}^{\text{int}}}\ket{sk}}{\sqrt{\sum_{s,k} |c_sb_{s, k}\alpha_{s,j,k}^{\text{int}}}|}
\ket{0^{a}}
\prod_{\ell=k}^{j}\left(\frac{\mathcal{T}e^{\int_{t_s + \ell h}^{t_s + (\ell+1)h} A(s) \mathrm{d}s} + \epsilon_1 \Lambda_{s,\ell}}{\alpha_{s,k}}\right)\ket{b(t_s + kh)} + \ket{\perp}\\
&=
\sum_j \frac{\ket{j}}{\sqrt{M_I}}
\frac{\sum_{s}\sum_{k<j}{c_s b_{s,k}}}{{\sum_{s}\sum_{k <j} |c_sb_{s, k}\alpha_{s,j,k}^{\text{int}}}|}
\ket{0^{a+d}}
\prod_{\ell = k}^{j-1}\left(\mathcal{T}e^{\int_{t_s + k h}^{t_s + (k+1)h} A(s) \mathrm{d}s} + \epsilon_1 \Lambda_{s,\ell}\right)\ket{b(t_s + kh)} + \ket{\perp}\\
&=
\sum_j 
\ket{j}\ket{0^{a+d_1}}
\frac{\sum_{s}\sum_{k<j}{c_s b_{s,k}}}{\sqrt{M_I}{\sum_{s}\sum_{k <j} |c_sb_{s, k}\alpha_{s,j,k}^{\text{int}}}|}
\left(\mathcal{T}e^{\int_{t_s + k h}^{t_s + jh} A(s) \mathrm{d}s} +  \Gamma_{s,j,k}\right)\ket{b(t_s + kh)} + \ket{\perp}.
\end{split}
\end{equation}
}

Notice that
\begin{equation}
\begin{split}
    \sum_s \sum_{k<j}c_s b_{s,k} \mathcal{T}e^{\int_{t_s + k h}^{t_s + (j-1)h} A(s) \mathrm{d}s} \ket{b(t_s+ kh)} &\approx u(jh),\\
    \left\|\sum_s \sum_{k<j}c_s b_{s,k} \mathcal{T}e^{\int_{t_s + k h}^{t_s + (j-1)h} A(s) \mathrm{d}s} \ket{b(t_s+ kh)}\right\| &\approx \|u(jh)\|,
\end{split}
\end{equation}
and
\begin{equation}
\begin{split}
    \sum_{s}\sum_{k<j} |c_s b_{s,k}| &= \mathcal{O}\left(\|u(0)\|+\int_{0}^{jh}|b(t)|\mathrm{d}t\right),\\
    \sum_{s}\sum_{k<j} |c_s b_{s,k} \alpha_{s,j,k}^{\text{int}}| &= \mathcal{O}\left(Q_0\left(\|u(0)\|+\int_{0}^{jh}|b(t)|\mathrm{d}t\right)\right)
\end{split}
\end{equation}
where $Q_0^{s,k,j}$ only depends on the norm of the propagator.

Then we can determine the repetition needed for amplitude amplification:
\begin{equation}
    1/Q := \left(\frac{1}{M_I}\sum_{j=0}^{M_I-1} \frac{\|u(jh)\|^2}{\left(Q_0^{s,k,j}\right)^2\left(\|u(0)\|+\int_{0}^{jh}|b(t)|\mathrm{d}t\right)^2}\right)^{1/2}.
\end{equation}
The error estimation is similar to previous estimation. The final complexity then should be
\begin{equation}
\label{eqn:time-marching-history-state-inhomo-complexity}
    \mathcal{O}\left(\alpha_A^2 Q T^2 \frac{\log^2(\alpha_A T Q/\epsilon_{\text{tol}})}{\log\log(\alpha_A TQ/\epsilon_{\text{tol}})}\text{Cost}(\text{HAM-T}_{A,q}) + Q\;\text{Cost}(O_{\text{init}})\right).
\end{equation}

\subsubsection{Dissipative inhomogeneous case}

Now we move on to the dissipative case.
Define the select oracle as
\begin{equation}
\begin{split}
\label{eqn:time-marching-history-dissi-select-oracle}
    O_{\text{sel}} &= 
    \sum_{s = 0}^{n-1}
    \sum_{j=0}^{M_I-1}
    \sum_{k=\max(0, j-r)}^{j-1}
    \ket{s}\bra{s}\otimes \ket{j} \bra{j} \otimes \ket{k}\bra{k}
    \otimes \mathcal{T}e^{\int_{t_s + kh}^{t_s +(j+1)h} A(s) \mathrm{d}s}\\
    &=
    \sum_{s = 0}^{n-1}
    \ket{s}\bra{s}\otimes 
    \sum_{j=0}^{M_I-1}
    \ket{j} \bra{j} \otimes 
    \prod_{k=j-r}^j
    \left[
    \left(\sum_{\ell = \max(0,j-r)}^{\max(-1,k)}
    \ket{\ell}\bra{\ell}\right)
    \otimes \mathcal{T}e^{\int_{t_s + k h}^{t_s +(k+1)h} A(s) \mathrm{d}s}\right]\\
\end{split}
\end{equation}
where $r$ is an integer to be determined.

In this case, we need to leverage a stronger input oracle defined in Equation~\eqref{eqn:stronger-input-model} to construct
\begin{equation}
\label{eqn:time-marching-his-dissi-part-construct}
    \sum_{j=0}^{M_I-1}
    \ket{j} \bra{j} \otimes 
    \prod_{k=j-r}^j
    \left[
    \left(\sum_{\ell = \max(0,j-r)}^{\max(-1,k)}
    \ket{\ell}\bra{\ell}\right)
    \otimes \mathcal{T}e^{\int_{t_s + k h}^{t_s +(k+1)h} A(s) \mathrm{d}s}\right].
\end{equation}
For $y = r, r-1, \cdots, 0$, it is clear that we can have 
\begin{equation}
    \sum_{j=y}^{M_I-1}\ket{j}\bra{j}\otimes 
    \left(\sum_{\ell = \max(0,j-r)}^{j-y}\ket{\ell}\bra{\ell}\right)
    \otimes
    \mathcal{T}e^{\int_{t_s + (j-y) h}^{t_s +(j+1-y)h} A(s) \mathrm{d}s}
\end{equation}
using only $\mathcal{O}((\alpha_A/\eta)\log(1/\epsilon_a)\log(1/\epsilon_0)/\log\log(1/\epsilon_0))$ queries to $\text{HAM-T}_{A}$. The $\alpha_A/\eta$ factor comes from the uniform amplitude amplification.
For a given $\ket{j}$ where $j \geq r$,
\begin{equation}
\begin{split}
    &\prod_{y=r}^0 \ket{j}\bra{j}\otimes 
    \left(\sum_{\ell = \max(0,j-r)}^{j-y}\ket{\ell}\bra{\ell}\right)
    \otimes
    \mathcal{T}e^{\int_{t_s + (j-y) h}^{t_s +(j+1-y)h} A(s) \mathrm{d}s}
    =
    \ket{j}\bra{j}\otimes 
    \prod_{y=r}^{0}\left(\sum_{\ell = \max(0,j-r)}^{j-y}\ket{\ell}\bra{\ell}\right)
    \otimes
    \mathcal{T}e^{\int_{t_s + (j-y) h}^{t_s +(j+1-y)h} A(s) \mathrm{d}s}\\
    &\stackrel{k = j-y}{=\joinrel=\joinrel=}
    \ket{j}\bra{j}\otimes 
    \prod_{k=j-r}^{j}\left(\sum_{\ell = j-r}^{k}\ket{\ell}\bra{\ell}\right)
    \otimes
    \mathcal{T}e^{\int_{t_s + k h}^{t_s +(k+1)h} A(s) \mathrm{d}s}.
\end{split}
\end{equation}
For those $j$'s that $j < r$,
\begin{equation}
\begin{split}
    &\prod_{y=j}^0 \ket{j}\bra{j}\otimes 
    \left(\sum_{\ell = \max(0,j-r)}^{j-y}\ket{\ell}\bra{\ell}\right)
    \otimes
    \mathcal{T}e^{\int_{t_s + (j-y) h}^{t_s +(j+1-y)h} A(s) \mathrm{d}s}
    \stackrel{k = j-y}{=\joinrel=\joinrel=}
    \ket{j}\bra{j}\otimes 
    \prod_{k=0}^{j}\left(\sum_{\ell = 0}^{k}\ket{\ell}\bra{\ell}\right)
    \otimes
    \mathcal{T}e^{\int_{t_s + k h}^{t_s +(k+1)h} A(s) \mathrm{d}s}.
\end{split}
\end{equation}
Using the standard Time-Marching method and the equations above, Equation~\eqref{eqn:time-marching-his-dissi-part-construct} is constructed, and the cost should be 
\begin{equation}
    \mathcal{O}\left(
    \frac{r\alpha_A}{\eta}\frac{\log(1/\epsilon_a)\log(1/\epsilon_0)}{\log\log(1/\epsilon_0)}
    \text{Cost}(\text{HAM-T}_A)
    \right),
\end{equation}
since there are $r$ propagators in a row.
Take the integration nodes into count, the select oracle defined in
Equation~\eqref{eqn:time-marching-history-dissi-select-oracle} still scales as
\begin{equation}
\label{eqn:time-marching-history-dissi-complexity-first}
    \mathcal{O}\left(
    \frac{r\alpha_A}{\eta}\frac{\log(1/\epsilon_a)\log(1/\epsilon_0)}{\log\log(1/\epsilon_0)}
    \text{Cost}(\text{HAM-T}_A)
    \right).
\end{equation}

Our next goal should be about determining $r$.
To make
\begin{equation}
    \left\|
    \left|
    \sum_{j=0}^{M_I-1}\ket{j}\bra{j} \otimes \widetilde{u(jh)}
    \right\rangle
    -
    \left|
    \sum_{j=0}^{M_I-1}\ket{j}\bra{j} \otimes u(jh)
    \right\rangle
    \right\|
    \leq \epsilon_{\text{truncate}},
\end{equation}
where
\begin{equation}
    \widetilde{u(jh)} = \mathbbm{1}_{jh < rh}\mathcal{T}e^{\int_0^{jh}A(s)\mathrm{d}s}u(0)
    + 
    \int_{\max(0,jh - rh)}^{jh}\mathcal{T}e^{\int_t^{jh}A(s)\mathrm{d}s}b(t)\mathrm{d}t,
\end{equation}
we can use Lemma~\ref{lem:error-ana}. For a given tolerance $\epsilon_{\text{truncate}}$,
we are requiring
\begin{equation}
    \frac{\sum_{j>r}\left\|
    \mathcal{T}e^{\int_0^{jh}A(s)\mathrm{d}s}u(0)
    +
    \int_0^{(j-r)h}\mathcal{T}e^{\int_t^{jh}A(s)\mathrm{d}s}b(t)\mathrm{d}t
    \right\|^2}{\sum_{j=0}^{M_I-1}\|u(jh)\|^2} \leq \epsilon_{\text{truncate}}^2/4.
\end{equation}

Since 
\begin{equation}
\begin{split}
    &\left\|
    \mathcal{T}e^{\int_0^{jh}A(s)\mathrm{d}s}u(0)
    +
    \int_0^{(j-r)h}\mathcal{T}e^{\int_t^{jh}A(s)\mathrm{d}s}b(t)\mathrm{d}t
    \right\|
    \leq
    \left\|
    \mathcal{T}e^{\int_0^{jh}A(s)\mathrm{d}s}u(0)
    \right\| +
    \left\|
    \int_0^{(j-r)h}\mathcal{T}e^{\int_t^{jh}A(s)\mathrm{d}s}b(t)\mathrm{d}t
    \right\|\\
    &\leq
    e^{-\eta jh}\|u(0)\|
    +
    \frac{e^{-\eta jh}\max_t \|b(t)\|}{\eta}(e^{\eta (j-r)h} - 1)
    \leq
    e^{-\eta rh}\left(\|u(0)\| + \frac{\max_t\|b(t)\|}{\eta}\right),
\end{split}
\end{equation}
we know that
\begin{equation}
\label{eqn:time-marching-dissi-bound-numerator}
    \sum_{j>r}\left\|
    \mathcal{T}e^{\int_0^{jh}A(s)\mathrm{d}s}u(0)
    +
    \int_0^{(j-r)h}\mathcal{T}e^{\int_t^{jh}A(s)\mathrm{d}s}b(t)\mathrm{d}t
    \right\|^2
    \leq
    M_I\left(\|u(0)\| + \frac{\max_t\|b(t)\|}{\eta}\right)^2 e^{-2\eta rh}.
\end{equation}

In the mean time,
\begin{equation}
\begin{split}
\label{eqn:time-marching-dissi-bound-denominator}
    \sum_{j=0}^{M_I-1}\|u(jh)\|^2
    &=
    \sum_{j=0}^{M_I-1}
    \left\|
    \mathcal{T}e^{\int_0^{jh}A(s)\mathrm{d}s}u(0)
    +
    \int_0^{jh}\mathcal{T}e^{\int_t^{jh}A(s)\mathrm{d}s}b(t)\mathrm{d}t
    \right\|^2\\
    &\geq
    \sum_{j=0}^{M_I-1}
    \left\|
    \int_0^{jh}\mathcal{T}e^{\int_t^{jh}A(s)\mathrm{d}s}b(t)\mathrm{d}t
    \right\|^2
    \geq \sum_{j=0}^{M_I-1}\left(\max_t \|b(t)\| \int_0^{jh} e^{-\alpha_A (jh-t)}\mathrm{d}t\right)^2\\
    &\geq (\max_t\|b(t)\|)^2\sum_{j=0}^{M-1}\frac{1-2e^{-\alpha_A jh}}{\alpha_A^2}
    =\max_t\|b(t)\|^2\frac{M_I - 2 + 2e^{-\alpha_A M h}}{\alpha_A^2(1 - e^{-\alpha_A h})^2}
    \geq \frac{(M_I-2)\max_t\|b(t)\|^2}{\alpha_A^2}.
\end{split}
\end{equation}
Combine Equation~\eqref{eqn:time-marching-dissi-bound-numerator} and Equation~\eqref{eqn:time-marching-dissi-bound-denominator} together, we can take $r$ as
\begin{equation}
    r = \Theta\left(
    \frac{\alpha_A}{\eta}\log\left(\frac{\alpha_A \left(\|u(0)\| + \max_t\|b(t)\|/\eta\right)}{\epsilon_{\text{truncate}}\max_t\|b(t)\|}\right)
    \right).
\end{equation}
Plug this into Equation~\eqref{eqn:time-marching-history-dissi-complexity-first},
the complexity of constructing the select oracle should be
\begin{equation}
    \mathcal{O}\left(
    \frac{\alpha_A^2}{\eta^2}\log\left(\frac{\alpha_A \left(\|u(0)\| + \max_t\|b(t)\|/\eta\right)}{\epsilon_{\text{truncate}}\max_t\|b(t)\|}\right)\frac{\log(1/\epsilon_a)\log(1/\epsilon_0)}{\log\log(1/\epsilon_0)}
    \text{Cost}(\text{HAM-T}_A)
    \right).
\end{equation}
With
\begin{equation}
    1/Q := \left(\frac{1}{M_I}\sum_{j=0}^{M_I-1} \frac{\|u(jh)\|^2}{\left(\|u(0)\|+\int_{0}^{jh}|b(t)|\mathrm{d}t\right)^2}\right)^{1/2},
\end{equation}
take
\begin{equation}
\begin{split}
    \epsilon_{\text{truncate}} = \mathcal{O}(\epsilon_{\text{tol}}), \quad
    &\epsilon_{\text{int}} = \mathcal{O}(\epsilon_{\text{tol}}/Q),\\
    \epsilon_{a} = \mathcal{O}(\epsilon_{\text{tol}}/(\alpha_A Q rh)) = \widetilde{\mathcal{O}}(\eta\epsilon_{\text{tol}}/(Q\alpha_A)), \quad
    &\epsilon_{0} = \mathcal{O}(\epsilon_{\text{tol}}/(\alpha_A Q rh))= \widetilde{\mathcal{O}}(\eta\epsilon_{\text{tol}}/(Q\alpha_A)),
\end{split}
\end{equation}
the complexity becomes 
\begin{equation}
\begin{split}
    \widetilde{\mathcal{O}}\bigg(
    \frac{\alpha_A^2 Q}{\eta^2}\log\left(\frac{\alpha_A \left(\|u(0)\| + \max_t\|b(t)\|/\eta\right)}{\epsilon_{\text{tol}}\max_t\|b(t)\|}\right)\frac{\log^2(\alpha_A Q /(\eta\epsilon_{\text{tol}}))}{\log\log(\alpha_A Q /(\eta\epsilon_{\text{tol}}))}
    \text{Cost}(\text{HAM-T}_A)
    +\\ 
    Q \; \text{Cost}(O_{\text{init}})
    \bigg).
\end{split}
\end{equation}

\section{Linear combination of Hamiltonian simulation}
\label{sec:lchs}
The LCHS (Linear Combination of Hamiltonian Simulation) algorithm ~\cite{an2023linear,an2023quantum} is the state-of-art quantum linear differential equation solver with
near-optimal dependence on all parameters. Here we firstly review its solution for final state, then we discuss the history state preparation.
The dissipative cases are also discussed.
\subsection{Final state preparation}
\subsubsection{Homogeneous case}
\label{sec:LCHS-Final-homo}
For the homogeneous linear differential equation
\begin{equation}
\begin{split}
    \frac{\mathrm{d}u(t)}{\mathrm{d}t} &= A(t) u(t),\\
    u(0) &= u_0,
\end{split}
\end{equation}
we can apply the Cartesian decomposition on $A(t)$, namely
\begin{equation}
\begin{split}
    A(t) &= L(t) + \imath H(t),\\
    L(t) = \frac{A(t) + A^\dag(t)}{2}&, \quad
    H(t) = \frac{A(t) - A^\dag(t)}{2\imath}.
\end{split}
\end{equation}
Then 
\begin{equation}
    \mathcal{T}e^{\int_0^T A(s)\mathrm{d}s}
    =
    \int_{\mathbb{R}}\frac{f(k)}{1 - \imath k}
    \mathcal{T}e^{\imath \int_0^T (kL(s) + H(s))\mathrm{d}s}\mathrm{d}k,
\end{equation}
where 
\begin{equation}
    f(k) = \frac{1}{2\pi e^{-2\beta}e^{(1+\imath z)^\beta}}, \quad \beta \in (0,1),
\end{equation}
provided with
\begin{equation}
    L(t) \leq 0.
\end{equation}

Define 
\begin{equation}
    U(T,t,k) = \mathcal{T}e^{\imath \int_t^T (kL(s) + H(s))\mathrm{d}s},\quad
    U(T,k) = U(T,0,k),
\end{equation}
we can discretize the integration into
\begin{equation}
\begin{split}
    \mathcal{T}e^{\int_0^T A(s)\mathrm{d}s} &= \int_{\mathbb{R}}g(k)U(T,k)\mathrm{d}k
    \approx \int_{-K}^K g(k)U(T,k)\mathrm{d}k\\ &= \sum_{m=-K/h_1}^{K/h_1 - 1}\int_{mh_1}^{(m+1)h_1}g(k)U(T,k)\mathrm{d}k \approx \sum_{m=-K/h_1}^{K/h_1-1}\sum_{q=0}^{Q-1} c_{q,m}U(T,k_{q,m}).
\end{split}
\end{equation}
We have the following theorem:
\begin{theorem}[Lemma 9 \& Lemma 10 \& Lemma 11 in~\cite{an2023quantum}]
    We are able to write
    \begin{equation}
        \left\|\mathcal{T}e^{\int_0^T A(s)\mathrm{d}s}
        -
        \sum_{m=-K/h_1}^{K/h_1-1}\sum_{q=0}^{Q-1} c_{q,m}U(T,k_{q,m})\right\|
        \leq \epsilon_{a},
    \end{equation}
    if 
    \begin{equation}
        K = \mathcal{O}\left(\left(\log^{1/\beta}(1/\epsilon_a)\right)\right), \quad
        h_1 = \frac{1}{eT \max_t \|L(t)\|}, \quad
        Q = \mathcal{O}\left(\log(1/\epsilon_a)\right).
    \end{equation}
    For simplicity, we can also write
    \begin{equation}
            \sum_{m=-K/h_1}^{K/h_1-1}\sum_{q=0}^{Q-1} c_{q,m}U(T,k_{q,m})
            =
            \sum_{j=0}^{M_s-1}c_j U(T,k_j).
    \end{equation}
    The number of unitaries in the summation is
    \begin{equation}
        M_s = \mathcal{O}\left(T\max_t\|L(t)\|\left(\log(1/\epsilon_a)\right)^{1+1/\beta}\right),
    \end{equation}
    and the coefficients satisfy
    \begin{equation}
        \sum_{q,m} c_{q,m} = \mathcal{O}(1).
    \end{equation}
\end{theorem}

Consider the preparation oracles
\begin{equation}
    \label{eqn:LCHS-homo-prepare-oracle}
    O_{c,l}: \ket{0} \rightarrow \frac{1}{\sqrt{\|c\|_1}}\sum_{j=0}^{M_s-1}\overline{\sqrt{c_j}}\ket{j}, \quad
    O_{c,r}: \ket{0} \rightarrow \frac{1}{\sqrt{\|c\|_1}}\sum_{j=0}^{M_s-1}{\sqrt{c_j}}\ket{j},
\end{equation}
where $\|c\|_1 = \sum_{j}|c_j|$,
and the select oracle
\begin{equation}
    S_H = \sum_{j=0}^{M_s-1}\ket{j}\bra{j}\otimes W_j,
\end{equation}
where $W_j$ block encodes $e^{\imath\int_0^T k_j L(s) + H(s)\mathrm{d}s}$ in an accuracy $\epsilon_0$.
Let us discuss the cost of the select oracle.
Use the truncated Dyson series method~\cite{low2018hamiltonian}, we are able to construct a $(1, \cdot, \epsilon_0)$ block-encoding of $e^{\imath \int_0^T \mathcal{H}(s)\mathrm{d}s}$ for any Hamiltonian $\mathcal{H}(t)$ by querying its HAM-T block-encoding for $\mathcal{O}(\alpha T\log(1/\epsilon_0)/\log\log(1/\epsilon_0))$ times.
Here $\alpha$ is the normalization factor, and $\epsilon_0$ is the accuracy.
In~\cite{an2023quantum}, a new input model $\text{HAM-T}_{kL+H,q}$ where
\begin{equation}
    \label{eqn:newer-input-oracle}
    (\bra{0}\otimes I)\text{HAM-T}_{kL+H,q}(\ket{0}\otimes I) = \sum_{j=0}^{M_s-1}\sum_{\ell=0}^{M_D-1} \ket{j}\bra{j}\otimes \ket{\ell}\bra{\ell} \otimes \frac{k_j L(qh + \ell h/M_D)+H(qh + \ell h/M_D)}{\alpha_L K + \alpha_H}
\end{equation}
can be implemented using $\mathcal{O}(1)$ queries to $\text{HAM-T}_{A,q}$. Here $\alpha_L \geq \max_t \|L(t)\|, \alpha_H \geq \max_t\|H(t)\|$, and $M_D$ is the number of nodes in each sub-interval.
Thus the select oracle $S_H$ can be achieved using
\begin{equation}
\label{eqn:lchs-homo-select-oracle-complexity}
    \mathcal{O}\left((\alpha_L K + \alpha_H)T \frac{\log(((\alpha_L K + \alpha_H)T)/\epsilon_0)}{\log\log(((\alpha_L K + \alpha_H)T)/\epsilon_0)}
    \text{Cost}(\text{HAM-T}_{A,q})
    \right).
\end{equation}

Since
\begin{equation}
\begin{split}
    \label{eqn:lchs-final-homo-lcu}
    &\left(O_{c,l}^\dag\otimes I \otimes I\right)
    S_H
    \left(O_{c,r}\otimes I \otimes I\right)\ket{0}\ket{0}O_u\ket{0^n}
    =
    \left(O_{c,l}^\dag\otimes I \otimes I\right)
    S_H
    \left(O_{c,r}\otimes I \otimes I\right)\ket{0}\ket{0}\ket{u_0}\\
    &= \left(O_{c,l}^\dag\otimes I \otimes I\right)
    S_H \frac{1}{\sqrt{\|c\|_1}}\left(\sum_{j=0}^{M_s-1}\sqrt{c_j}\ket{j}\right)\ket{0}\ket{u_0}=
    \left(O_{c,l}^\dag\otimes I \otimes I\right)
    \frac{1}{\sqrt{\|c\|_1}}\left(\sum_{j=0}^{M_s-1}\sqrt{c_j}\ket{j}\right)W_j\ket{0}\ket{u_0}\\
    &= \left(O_{c,l}^\dag\otimes I \otimes I\right)
    \frac{1}{\sqrt{\|c\|_1}} \sum_{j=0}^{M_s-1}\sqrt{c_j}\ket{j}\ket{0}(U(T,k_j) + \epsilon_0 \Lambda_j)\ket{u_0} + \ket{\perp}\\
    &=
    \frac{1}{\|c\|_1}\ket{0}\ket{0}\left(\sum_{j=0}^{M_s-1}c_j (U(T,k_j) + \epsilon_0 \Lambda_j)\right)\ket{u_0} + \ket{\perp},
\end{split}
\end{equation}
given a tolerance $\epsilon_{\text{tol}}$, we require that
\begin{equation}
    \frac{\|u(0)\|}{\|u(T)\|}\sum_{j=0}^{M_s-1}|c_j|\epsilon_0 \leq \epsilon_{\text{tol}}/2,\quad
    \frac{\|u(0)\|}{\|u(T)\|}\epsilon_a \leq \epsilon_{\text{tol}}/2,
\end{equation}
indicating that we need to take $\epsilon_0 = \mathcal{O}(\|u(T)\|\epsilon_{\text{tol}}/\|u(0)\|)$, $\epsilon_a = \mathcal{O}(\|u(T)\|\epsilon_{\text{tol}}/\|u(0)\|)$.
Plugging this into Equation~\eqref{eqn:lchs-homo-select-oracle-complexity}, we observe that the final complexity should be
\begin{equation}
    \widetilde{\mathcal{O}}\left(\frac{\|u(0)\|}{\|u(T)\|}\alpha_A T \log^{1+1/\beta}(1/\epsilon_{\text{tol}})\;
    \text{Cost}(\text{HAM-T}_{A,q})
    + 
    \frac{\|u(0)\|}{\|u(T)\|}\;\text{Cost}(O_{\text{init}})\right).
\end{equation}

\subsubsection{Inhomogeneous case}
\label{sec:LCHS-final-inhomo}
For the inhomogeneous case,
\begin{equation}
\begin{split}
    \frac{\mathrm{d}u(t)}{\mathrm{d}t} &= A(t)u(t) + b(t),\\
    u(0) &= u_0,
\end{split}
\end{equation}
the final solution at time $T$ is 
\begin{equation}
    u(T) = \mathcal{T}e^{\int_0^T A(s)\mathrm{d}s}u(0)
    + \int_0^T \mathcal{T}e^{\int_t^T A(s)\mathrm{d}s}b(t) \mathrm{d}t,
\end{equation}
where the second part can be approximated by
\begin{equation}
\begin{split}
    \int_0^T \mathcal{T}e^{\int_t^T A(s)\mathrm{d}s}b(t) \mathrm{d}t &\approx \int_0^T \sum_{m_1 = -K/h_1}^{K/h_1 - 1}\sum_{q_1 = 0}^{Q_1-1}c_{q_1, m_1}U(T, t, k_{q_1,m_1})b(t)\mathrm{d}t\\
    &\approx \sum_{m_2=0}^{T/h_2 - 1}\sum_{q_2 = 0}^{Q_2 -1}\sum_{m_1 = -K/h_1}^{K/h_1 - 1}\sum_{q_1 = 0}^{Q_1-1} c'_{q_2, m_2}c_{q_1,m_1}U(T, t_{q_2,m_2}, k_{q_1,m_1})\ket{b(t_{q_2,m_2})}.
\end{split}
\end{equation}
The $c'_{q_2,m_2}$ is the corresponding Gaussian weights multiplied by $\|b(t_{q_2, m_2})\|$.
Similarly, we may bound the error caused by the numerical integration:
\begin{theorem}[Lemma 12 \& Lemma 13 in~\cite{an2023quantum}]
    We are able to write
    \begin{equation}
        \left\|\int_0^T \mathcal{T}e^{\int_t^T A(s)\mathrm{d}s}b(t)\mathrm{d}t
        -
        \sum_{m_2=0}^{T/h_2 - 1}\sum_{q_2 = 0}^{Q_2 -1}\sum_{m_1 = -K/h_1}^{K/h_1 - 1}\sum_{q_1 = 0}^{Q_1-1} c'_{q_2, m_2}c_{q_1,m_1}U(T, t_{q_2,m_2}, k_{q_1,m_1})\ket{b(t_{q_2,m_2})}\right\|
        \leq \epsilon_b
    \end{equation}
    if 
    \begin{equation}
    \begin{split}
        &K = \mathcal{O}\left(\log^{1/\beta}(1 + \|b\|_{L_1}/\epsilon_b)\right), \quad
        h_1 = \frac{1}{eT \max_t \|L(t)\|}, \quad
        h_2 = \frac{1}{eK\left(\sup_{p\geq 0}\|A^{(p)}\|^{1/(p+1)}+\sup_{p\geq 0}\|b^{(p)}\|^{1/(p+1)}\right)},\\
        &Q_1 = \mathcal{O}\left(\log(1 + \|b\|_{L_1}/\epsilon_b)\right), \quad
        Q_2 = \mathcal{O}\left(\log\left(T\left(\sup_{p\geq 0}\|A^{(p)}\|^{1/(p+1)}+\sup_{p\geq 0}\|b^{(p)}\|^{1/(p+1)}\right)\bigg/\epsilon_b\right)\right).
    \end{split}
    \end{equation}
    With some replacement of the indices,
    \begin{equation}
        u(T) \approx \sum_{j=0}^{M_s-1} c_j U(T,k_j)\ket{u_0}
        + \sum_{j'=0}^{M'_s-1}\sum_{j=0}^{M_s-1}c_j'c_j U(T, t_{j'}, k_j)\ket{b(t_{j'})}.
    \end{equation}
    The number of unitaries in the summation is
    \begin{equation}
        M_s' = \frac{TQ_2}{h_1} = \mathcal{O}\left(T\left(\sup_{p\geq 0}\|A^{(p)}\|^{1/(p+1)}+\sup_{p\geq 0}\|b^{(p)}\|^{1/(p+1)}\right)\left(\log^{1/\beta}(1+\|b\|_{L_1}/\epsilon_b)\right)\log(1/\epsilon_b)\right),
    \end{equation}
    and the coefficients satisfy
    \begin{equation}
        \sum_{q_2,m_2} c'_{q_2,m_2} = \mathcal{O}(\|b\|_{L_1}).
    \end{equation}
\end{theorem}

Consider the select oracle
\begin{equation}
    S_{IH} = \sum_{j'=0}^{M_s'-1}\sum_{j=0}^{M_s}\ket{j'}\bra{j'}\otimes \ket{j}\bra{j} \otimes W_{j,j'},
\end{equation}
where $W_{j,j'}$ is the block-encoding of $V_{j,j'} = U(T, t_{j'}, k_j)$. Notice that
\begin{equation}
    \label{eqn:indicator-compare}
    U(T,t,k) = \mathcal{T}e^{\imath \int_t^T kL(s) + H(s)\mathrm{d}s} = \mathcal{T}e^{\imath \int_0^T \mathbbm{1}_{s \geq t} (k L(s) + H(s))\mathrm{d}s}.
\end{equation}
We are thus able to implement the select oracle $S_{IH}$ using
\begin{equation}
    \widetilde{\mathcal{O}}\left(\alpha_A T \left(\log^{1/\beta}(1 + \|b\|_{L_1}/\epsilon_b)\right) \log(1/\epsilon_0)\right)
\end{equation}
times of $\text{HAM-T}_{kL+H,q}$ as defined in Equation~\eqref{eqn:newer-input-oracle}. 
Similarly, the select oracle $S_H$ costs
\begin{equation}
    \widetilde{\mathcal{O}}\left(\alpha_A T \left(\log^{1/\beta}(1/\epsilon_a)\right) \log(1/\epsilon_0)\right).
\end{equation}

Let us first use LCU to construct the inhomogeneous term. To this end,
define
\begin{equation}
    O_{c',l}: \ket{0} \rightarrow \frac{1}{\sqrt{\|c'\|_1}}\sum_{j'=0}^{M_s'-1}\overline{\sqrt{c_j'}}\ket{j'}, \quad
    O_{c',r}: \ket{0} \rightarrow \frac{1}{\sqrt{\|c'\|_1}}\sum_{j'=0}^{M_s'-1}{\sqrt{c_j'}}\ket{j'}.
\end{equation}
Then, we have
\begin{equation}
\begin{split}
    \label{eqn:lchs-final-inhomo-lcu}
    &\left(O_{c',l}^\dag\otimes O_{c,l}^\dag \otimes I \otimes I\right)
    S_{IH} ( I \otimes O_{\text{init}})
    \left(O_{c',r}\otimes O_{c,r} \otimes I \otimes I \right)\ket{0}\ket{0}\ket{0}\ket{0^n}\\
    &=
    \left(O_{c',l}^\dag\otimes O_{c,l}^\dag \otimes I \otimes I\right)
    S_{IH} ( I \otimes O_{\text{init}})
    \frac{1}{\sqrt{\|c\|_1\|c'\|_1}}\sum_{j'}\sum_{j}\sqrt{c_j c_j'}\ket{j'}\ket{j}\ket{0}\ket{0^n}\\
    &=
    \left(O_{c',l}^\dag\otimes O_{c,l}^\dag \otimes I \otimes I\right)
    S_{IH} 
    \frac{1}{\sqrt{\|c\|_1\|c'\|_1}}\sum_{j'}\sum_{j}\sqrt{c_j c_j'}\ket{j'}\ket{j}\ket{0}\ket{b(t_{j'})}\\
    &=
    \left(O_{c',l}^\dag\otimes O_{c,l}^\dag \otimes I \otimes I\right)
    \frac{1}{\sqrt{\|c\|_1\|c'\|_1}}\sum_{j'}\sum_{j}\sqrt{c_j c_j'}\ket{j'}\ket{j}\ket{0}(U(T,t_{j'},k_j) + \epsilon_0 \Lambda_{j,j'})\ket{b(t_{j'})} + \ket{\perp}\\
    &=
    \left(O_{c',l}^\dag\otimes O_{c,l}^\dag \otimes I \otimes I\right)
    \frac{1}{\sqrt{\|c\|_1\|c'\|_1}}\sum_{j'}\sum_{j}\sqrt{c_j c_j'}\ket{j'}\ket{j}\ket{0}(U(T,t_{j'},k_j) + \epsilon_0 \Lambda_{j,j'})\ket{b(t_{j'})} + \ket{\perp}\\
    &=
    \frac{1}{\|c\|_1\|c'\|_1}\sum_{j'}\sum_{j}{c_j c_j'}\ket{0}\ket{0}\ket{0}(U(T,t_{j'},k_j) + \epsilon_0 \Lambda_{j,j'})\ket{b(t_{j'})} + \ket{\perp}.
\end{split}
\end{equation}

Now, we combine Equation~\eqref{eqn:lchs-final-homo-lcu} and Equation~\eqref{eqn:lchs-final-inhomo-lcu} together.
Now, consider a rotation gate
\begin{equation}
    \label{eqn:rotation-h-ih}
    R: \ket{0} \rightarrow \frac{1}{\sqrt{\|u_0\| + \|c'\|_1}}(\sqrt{\|u_0\|} \ket{0} + \sqrt{\|c'\|_1}\ket{1}),
\end{equation}
and we prepare
\begin{equation}
\begin{split}
    \frac{1}{\sqrt{\|u_0\| + \|c'\|_1}}\bigg[\sqrt{\|u_0\|} \ket{0}
    \frac{1}{\|c\|_1}\ket{0}\ket{0}\left(\sum_{j=0}^{M_s-1}c_j (U(T,k_j) + \epsilon_0 \Lambda_j)\right)\ket{u_0}
    +\\
    \sqrt{\|c'\|_1}\ket{1}
    \frac{1}{\|c\|_1\|c'\|_1}\sum_{j'}\sum_{j}{c_j c_j'}\ket{0}\ket{0}\ket{0}(U(T,t_{j'},k_j) + \epsilon_0 \Lambda_{j,j'})\ket{b(t_{j'})}\bigg] + \ket{\perp}.
\end{split}
\end{equation}
Applying $R^\dag$, we obtain
\begin{equation}
    \label{eqn:LCHS-final-inhomo-before-aa}
    \frac{1}{\|c\|_1 (\|u_0\| + \|c'\|_{1})}\left(
    \sum_{j=0}^{M_s-1} c_j (U(T,0,k_j) + \epsilon_0 \Lambda_j) \|u_0\|\ket{u_0} + \sum_{j'=0}^{M_s'-1}\sum_{j=0}^{M_s}c_{j'}c_j (U(T,t_{j'},k_j) + \epsilon_0 \Lambda_{j,j'})\ket{b(t_{j'})}
    \right).
\end{equation}

Require that
\begin{equation}
    \frac{\epsilon_b + \|u_0\| \epsilon_a}{\|u(T)\|} \leq \epsilon_{\text{tol}}/2, \quad
    \frac{\|b\|_{L_1} + \|u_0\|}{\|u(T)\|}\|c\|_{1}\epsilon_0 \leq \epsilon_{\text{tol}}/2, 
\end{equation}
we take 
\begin{equation}
    \epsilon_a = \mathcal{O}\left(\frac{\|u(T)\|}{\|u_0\|} \epsilon_{\text{tol}}\right), \quad
    \epsilon_b = \mathcal{O}\left(\|u(T)\| \epsilon_{\text{tol}}\right), \quad
    \epsilon_0 = \mathcal{O}\left(\frac{\|u(T)\|}{\|b\|_{L_1} + \|u_0\|} \epsilon_{\text{tol}}\right).
\end{equation}
Thus, the complexity then becomes
\begin{equation}
    \widetilde{\mathcal{O}}\left(\frac{\|u(0)\| + \|b\|_{L_1}}{\|u(T)\|}\alpha_A T \log^{1+1/\beta}(1/\epsilon_{\text{tol}})\;
    \text{Cost}(\text{HAM-T}_{A,q})
    + 
    \frac{\|u(0)\| + \|b\|_{L_1}}{\|u(T)\|}\;\text{Cost}(O_{\text{init}})\right).
\end{equation}

\subsubsection{Dissipative inhomogeneous case}

Now we may consider the dissipative case. We may just use the analysis in Section~\ref{section:time-marching-inhomo-dissipative}, which indicates that
for a given error $\epsilon_{\text{truncate}}$ and the dissipative dynamics namely $A(t) + A^\dag(t) \leq -2\eta <0$,
we only need to simulate the dynamics over $[T-T_0,T]$ with  
\begin{equation}
    T_0 = \mathcal{O}\left(\log\left(\frac{\max_t \|b(t)\|}{\eta \epsilon_{\text{truncate}}\|u(T)\|}\right)\bigg/ \eta\right).
\end{equation}
It is obvious that we can take $\epsilon_{\text{truncate}} = \mathcal{O}\left(\frac{\|u(T)\|}{\|b\|_{L_1} + \|u_0\|}\epsilon_{\text{tol}}\right)$, and the complexity then becomes
\begin{equation}
    \widetilde{\mathcal{O}}\left(\frac{\|u(0)\| + \|b\|_{L_1}}{\|u(T)\|}(\alpha_A/\eta) \log^{2+1/\beta}(1/\epsilon_{\text{tol}})\;
    \text{Cost}(\text{HAM-T}_{A,q})
    + 
    \frac{\|u(0)\| + \|b\|_{L_1}}{\|u(T)\|}\;\text{Cost}(O_{\text{init}})\right).
\end{equation}

\subsection{History state preparation}
\subsubsection{Homogeneous case}
\label{section:LCHS-His-Homo}
Similar to that for the Time-Marching method, we are now trying to prepare a quantum state proportional to 
\begin{equation}
    \sum_{r=0}^{M-1}\ket{r}\bra{r}\otimes u(rh). 
\end{equation}
Then we need the operator
\begin{equation}
    \sum_{r=0}^{M-1}\ket{r}\bra{r} \otimes \mathcal{T}e^{\int_0^{rh}A(t)\mathrm{d}t} = 
    \sum_{r=0}^{M-1}\ket{r}\bra{r} \otimes \mathcal{T}e^{\int_0^{T}\mathbbm{1}_{t \leq rh}A(t)\mathrm{d}t},
\end{equation}
which can be implemented by the select oracle
\begin{equation}
    S_{H}^{\text{History}}=\sum_{r=0}^{M-1}\sum_{j=0}^{M_s-1}\ket{r}\bra{r}\otimes\ket{j}\bra{j} \otimes W_{r,j}, 
\end{equation}
and preparation oracles defined in Equation~\eqref{eqn:LCHS-homo-prepare-oracle}.
The $W_{r,j}$'s are block-encodings of $U(rh, k_j)$.
Fortunately, for any index $j$, $\sum_{r=0}^{M-1} \ket{r}\bra{r}\otimes U(rh,0,k_j)$ can be constructed in the same time, 
due to the fact that 
\begin{equation}
    U(rh,0,k_j) = \mathcal{T}e^{\imath \int_0^{rh} k_jL(s) + H(s)\mathrm{d}s} = \mathcal{T}e^{\imath \int_0^{T} \mathbbm{1}_{s \leq rh} (k_j L(s) + H(s))\mathrm{d}s}.
\end{equation}
Previous discussion shows that for any $r$, $U(rh,0,k_j)$ can also be constructed in parallel; thus, the select oracle costs
\begin{equation}
    \widetilde{\mathcal{O}}\left(\alpha_A T \log^{1/\beta}(1/\epsilon_a)\log(1/\epsilon_0)\right).
\end{equation}

According to Section~\ref{sec:LCHS-Final-homo}, what we truly implement acting on the state $\sum_{r=0}^{M-1}\ket{r}\bra{r}\otimes\frac{\ket{u_0}}{\sqrt{M}}$ should be like
\begin{equation}
    \ket{0}\sum_{r=0}^{M-1}\ket{r}\bra{r}\otimes \frac{1}{\sqrt{M}\|c\|_1}\left(\sum_{j=0}^{M_s-1}c_j \left(U(rh, k_j) + \epsilon_0 \Lambda_{rj}\right)\right)\ket{u_0} + \ket{\perp},
\end{equation}
and the cost for now is
\begin{equation}
    \label{eqn:lchs-history-homo-first}
    \widetilde{\mathcal{O}}\left(\alpha_AT\log(1/\epsilon_0)\log^{1/\beta}(1/\epsilon_{a})\text{Cost}(\text{HAM-T}_{A,q}) + \text{Cost}(O_{\text{init}})\right).
\end{equation}

Given an error tolerance $\epsilon_{\text{tol}}$, we need to restrict
\begin{equation}
    \frac{\sqrt{\sum_{r=0}^{M-1}\epsilon_a^2}}{\sqrt{\sum_{r=0}^{M-1}\|u(rh)\|^2/\|u(0)\|^2}} \leq \frac{\epsilon_{\text{tol}}}{2}, \quad
    \frac{\sqrt{\sum_{r=0}^{M-1}\sum_{j=0}^{M_s-1}c_j^2\epsilon_0^2}}{\sqrt{\sum_{r=0}^{M-1}\|u(rh)\|^2/\|u(0)\|^2}} \leq \frac{\epsilon_{\text{tol}}}{2},
\end{equation}
meaning we may choose
\begin{equation}
    \epsilon_a = \mathcal{O}\left(\sqrt{\frac{\sum_{r=0}^{M-1}\|u(rh)\|^2/\|u(0)\|^2}{M}}\epsilon_{\text{tol}}\right),\quad
    \epsilon_0 = \mathcal{O}\left(\sqrt{\frac{\sum_{r=0}^{M-1}\|u(rh)\|^2/\|u(0)\|^2}{M}}\epsilon_{\text{tol}}\right).
\end{equation}
Plug the choices of $\epsilon_a$ and $\epsilon_0$ into Equation~\eqref{eqn:lchs-history-homo-first} and apply the amplitude amplification, the cost is
\begin{equation}
    \widetilde{\mathcal{O}}\left(\frac{\|u(0)\|}{\sqrt{\sum_{r=0}^{M-1}\|u(rh)\|^2/M}}\alpha_AT\log^{1+1/\beta}(1/\epsilon_{\text{tol}})\text{Cost}(\text{HAM-T}_{A,q}) + \frac{\|u(0)\|}{\sqrt{\sum_{r=0}^{M-1}\|u(rh)\|^2/M}}\text{Cost}(O_{\text{init}})\right).
\end{equation}

\subsubsection{Dissipative homogeneous case}

For the dissipative case, Equation~\eqref{eqn:homo-history-dissi-time-choice} tells us for a given truncation error $\epsilon_{\text{truncate}}$, we only need to simulate 
\begin{equation}
    T_0 = \mathcal{O}\left(\log(1/\epsilon_{\text{truncate}})/\eta\right).
\end{equation}
We only need to take
\begin{equation}
    \epsilon_{\text{truncate}} = \mathcal{O}\left(\frac{\sqrt{\sum_{r=0}^{M-1}\|u(rh)\|^2/M}}{\|u(0)\|}\epsilon_{\text{tol}}\right),
\end{equation}
and keep the choice of $\epsilon_a$ and $\epsilon_0$ (in the big-$O$ sense). The resulting complexity is

\begin{equation}
    \widetilde{\mathcal{O}}\left(\frac{\|u(0)\|}{\sqrt{\sum_{r=0}^{M-1}\|u(rh)\|^2/M}}(\alpha_A/\eta)\log^{2+1/\beta}(1/\epsilon_{\text{tol}})\text{Cost}(\text{HAM-T}_{A,q}) + \frac{\|u(0)\|}{\sqrt{\sum_{r=0}^{M-1}\|u(rh)\|^2/M}}\text{Cost}(O_{\text{init}})\right).
\end{equation}

\subsubsection{Inhomogeneous case}
For any $rh$, the state is
\begin{equation}
    u(rh) = \mathcal{T}e^{\int_0^{rh} A(t)\mathrm{d}t}u(0) + \int_0^{rh}\mathcal{T}e^{\int_t^{rh} A(s)\mathrm{d}s}b(t)\mathrm{d}t.
\end{equation}
Let us define the operator implemented in Equation~\eqref{eqn:LCHS-final-inhomo-before-aa} as $P(T)$,
we may define a select oracle as follows:
\begin{equation}
\begin{split}
    S_{IH}^{\text{History}}
    = \sum_{r=0}^{M-1} \sum_{j' = 0}^{M_s' -1}\sum_{j=0}^{M_s} \ket{r}\bra{r} \otimes \ket{j'}\bra{j'} \otimes \ket{j}\bra{j} \otimes W_{r,j,j'},
\end{split}
\end{equation}
where $W_{r,j,j'}$ is the block-encoding of $\mathbbm{1}_{t_j'\leq rh}U(rh, t_j', k_j) + (1 - \mathbbm{1}_{t_j' \leq rh})I$.
Use the same argument in Equation~\eqref{eqn:indicator-compare}, it is exactly
\begin{equation}
    e^{\imath \int_{t_j'}^rh kL(s) + H(s)\mathrm{d}s}
    = e^{\imath \int_{0}^T \mathbbm{1}_{t_j' \leq s \leq rh}(kL(s) + H(s))\mathrm{d}s}.
\end{equation}
Thus we need
\begin{equation}
    \widetilde{\mathcal{O}}\left(
        \alpha_A \log^{1/\beta}(1 + \|b\|_{L_1}/\epsilon_b)\log(1/\epsilon_0)
    \right)
\end{equation}
queries of HAM-T$_{A,q}$ to implement $S_{IH}^{\text{History}}$.
Consider preparation oracles
\begin{equation}
\begin{split}
    O_{c,l}^{\text{History}}\left(\frac{1}{\sqrt{M}}\sum_r \ket{r}\right) \ket{0}\ket{0} &= \frac{1}{\sqrt{M}}\sum_r \ket{r}\frac{1}{\sqrt{\sum_{j,j'} |c_j c_{j'}^r|}}\sum_{j,j'}\overline{\sqrt{c_j c_{j'}^r}}\ket{j}\ket{j'},\\
    O_{c,r}^{\text{History}}\left(\frac{1}{\sqrt{M}}\sum_r \ket{r}\right) \ket{0}\ket{0} &= \frac{1}{\sqrt{M}}\sum_r \ket{r}\frac{1}{\sqrt{\sum_{j,j'} |c_j c_{j'}^r|}}\sum_{j,j'}{\sqrt{c_j c_{j'}^r}}\ket{j}\ket{j'},\\
\end{split}
\end{equation}
Here $c_{j'}^r = 0$ if $t_{j'} > rh$.

Apply
\begin{equation}
    (O_{c,l}^{\text{History}})^\dag S_{IH}^{\text{History}} O_{\text{init}} O_{c,r}^{\text{History}}
\end{equation}
on the initial state $\frac{1}{\sqrt{M}}\sum_{r=0}^{M-1}\ket{r}\otimes \ket{0}\ket{0} \ket{0}$, 
also prepare the history state just as described in Section~\ref{section:LCHS-His-Homo},
use the rotation operator defined as Equation~\eqref{eqn:rotation-h-ih} controlled on every $\ket{r}$,
what we have is
\begin{equation}
\begin{split}
    &\ket{0}\frac{1}{\sqrt{M}}\sum_{r=0}^{M-1}\ket{r}\bra{r} \otimes\\ 
    &\frac{1}{\|c\|_1 (\|u_0\| + \|c_r'\|_{1})}\left(
    \sum_{j=0}^{M_s-1} c_j (U(rh,0,k_j) + \epsilon_0 \Lambda_j) \|u_0\|\ket{u_0} +
    \sum_{j'=0}^{M_s-1}\sum_{j=0}^{M_s}c_{j'}c_j (U(rh,t_{j'},k_j) + \epsilon_0 \Lambda_{j,j'})\ket{b(t_{j'})}
    \right)
    + \ket{\perp}.
\end{split}
\end{equation}
Here $\|c_r'\|_1 = \sum_{j'} |c_{j'}^r| = \mathcal{O}\left(\int_0^{rh}|b(t)|\mathrm{d}t\right)$.

It is rather clear that the implementation costs
\begin{equation}
    \widetilde{\mathcal{O}}\left(\alpha_A \left(\log^{1/\beta}(1 + \|b\|_{L_1}/\epsilon_b) + \log^{1/\beta}(1/\epsilon_a)\right) T \log(1/\epsilon_0)\right).
\end{equation}
We need to choose
\begin{equation}
\begin{split}
    \left(\sum_{r=0}^{M-1}\frac{\left(\|u_0\| + \int_0^{rh}|b(t)|\mathrm{d}t\right)^2 \epsilon_0^2}{\|u(rh)\|^2}\right)^{1/2}&\leq \frac{\epsilon_{\text{tol}}}{3},\\
    \left(\frac{\sum_{r=0}^{M-1} \epsilon_a^2}{\sum_{r=0}^{M-1}\|u(rh)\|^2}\right)^{1/2}&\leq \frac{\epsilon_{\text{tol}}}{3},\\
    \left(\frac{\sum_{r=0}^{M-1} \epsilon_b^2}{\sum_{r=0}^{M-1}\|u(rh)\|^2}\right)^{1/2}&\leq \frac{\epsilon_{\text{tol}}}{3},
\end{split}
\end{equation}
indicating that the cost should be
\begin{equation}
\begin{split}
    \widetilde{\mathcal{O}}\bigg(\left(\frac{1}{\sum_{r=0}^{M-1}\frac{\|u(rh)\|^2/M}{\left(\|u_0\| + \int_0^{rh}|b(t)|\mathrm{d}t\right)^2 }}\right)^{1/2}\alpha_A T \log^{1+1/\beta}(1/\epsilon_{\text{tol}})\;
    \text{Cost}(\text{HAM-T}_{A,q})
    +\\ 
    \left(\frac{1}{\sum_{r=0}^{M-1}\frac{\|u(rh)\|^2/M}{\left(\|u_0\| + \int_0^{rh}|b(t)|\mathrm{d}t\right)^2 }}\right)^{1/2}\;\text{Cost}(O_{\text{init}})
    \bigg).
\end{split}
\end{equation}

\subsubsection{Dissipative inhomogeneous case}

For the dissipative case, we need to determine an integer $w$ such that for the select oracle
\begin{equation}
\begin{split}
    \sum_{r=0}^{M-1} \sum_{j' = 0}^{M_s' -1}\sum_{j=0}^{M_s} \ket{r}\bra{r} \otimes \ket{j'}\bra{j'} \otimes \ket{j}\bra{j} \otimes W_{r,j,j',w},
\end{split}
\end{equation}
where $W_{r,j,j',w} = \mathbbm{1}_{(r-w) h \leq t_j' \leq rh}U(rh,t_j',k_j)$.

By adopting the stronger input model in Section~\ref{section:input-models}, we can implement HAM-T$_{kL + H}$ defined as
\begin{equation}
    (\bra{0}\otimes I)\text{HAM-T}_{kL+H}(\ket{0}\otimes I)
    =
    \sum_{r=0}^{M-1}\sum_{j=0}^{M_s-1}\sum_{\ell =0}^{M_D-1}\ket{r}\bra{r} \otimes \ket{j}\bra{j} \otimes \ket{\ell}\bra{\ell} \otimes \frac{k_j L(rh + \ell h/M_D) + H(rh + \ell h/M_D)}{\alpha_L K + \alpha_H}
\end{equation}
using only $\mathcal{O}(1)$ quries of HAM-T$_{A}$. Thus the select oracle defined above
combined with another select oracle of the homogeneous part only cost
\begin{equation}
    \widetilde{\mathcal{O}}\left(\left(\log^{1/\beta}(1 + \|b\|_{L_1}/\epsilon_b) + \log^{1/\beta}(1/\epsilon_a)\right) (w) \log(1/\epsilon_0)\text{Cost}(\text{HAM-T}_A)\right).
\end{equation}
Section~\ref{section:time-marching-his-inhomo} tells us that 
\begin{equation}
    w  = \Theta\left(
    \frac{\alpha_A}{\eta}\log\left(\frac{\alpha_A \left(\|u(0)\| + \max_t\|b(t)\|/\eta\right)}{\epsilon_{\text{truncate}}\max_t\|b(t)\|}\right)
    \right).
\end{equation}
The error analysis is similar, where we could choose certain errors to be sufficiently small without introducing overhead beyond poly-logarithmic dependence. 
The final complexity is given as 
\begin{equation}
\begin{split}
    \widetilde{\mathcal{O}}\bigg(\left(\frac{1}{\sum_{r=0}^{M-1}\frac{\|u(rh)\|^2/M}{\left(\|u_0\| + \int_0^{rh}|b(t)|\mathrm{d}t\right)^2 }}\right)^{1/2}\frac{\alpha_A}{\eta} \log^{2+1/\beta}(1/\epsilon_{\text{tol}})\;
    \text{Cost}(\text{HAM-T}_{A,q})
    +\\ 
    \left(\frac{1}{\sum_{r=0}^{M-1}\frac{\|u(rh)\|^2/M}{\left(\|u_0\| + \int_0^{rh}|b(t)|\mathrm{d}t\right)^2 }}\right)^{1/2}\;\text{Cost}(O_{\text{init}})
    \bigg).
\end{split}
\end{equation}

\section{Applications}\label{sec:applications}

In this section, we present two concrete examples of dissipative ODEs and show that we can apply our fast-forwarded quantum algorithms to achieve improved complexity.

\subsection{Quantum dynamics with non-Hermitian Hamiltonians}

Quantum dynamics with a non-Hermitian Hamiltonian has been a vividly and vastly emerging research area over the past few years. 
It provides a more precise and comprehensive model for quantum systems interacting with environments and has found a huge number of applications, such as open quantum systems, quantum resonances, quantum transport, field theories, and quantum many-body systems~\cite{Bender2007,RegoMonteiroNobre2013,GiusteriMattiottiCelardo2015,ElGanainyMakrisKhajavikhanEtAl2018,GongAshidaKawabataEtAl2018,OkumaKawabataShiozakiEtAl2020,AshidaGongUeda2020,MatsumotoKawabataAshidaEtAl2020,DingFangMa2022,ChenSongLado2023,ZhengQiaoWangEtAl2024,ShenLuLadoEtAl2024}. 
The dynamics is described by a time-dependent Schr\"odinger equation with a non-Hermitian Hamiltonian, given as 
\begin{equation}\label{eqn:app_non_Hermitian_quantum}
    i \frac{du(t)}{dt} = (H(t) + i L(t)) u(t). 
\end{equation}
Here $L(t)$ and $H(t)$ are two time-dependent Hermitian matrices. 
$H(t)$ is the Hamiltonian for the closed system, and $L(t)$ represents the so-called ``non-Hermitian'' part. 
We assume $L(t) \leq -2\eta < 0$ for all $t$, so the dynamics is dissipative. 
According to Result~\ref{result:dissi}, our algorithms can prepare an $\epsilon$ approximation of the history state over $[0,T]$ with total query complexity $\mathcal{O}(\log^3(1/\epsilon))$, independent of time $T$, and a low state preparation cost. 
Notice that there is no inhomogeneous term in Equation~\eqref{eqn:app_non_Hermitian_quantum}, so its final state preparation for long time is trivial.

\subsection{Reaction-diffusion process}

The reaction-diffusion process describes the phenomenon in which reaction (or advection) and diffusion simultaneously exist. 
This process appears ubiquitously in fluid dynamics, chemical processes, time evolution of biological species, ecological networks, and population growth, to name a few~\cite{evans2010partial,hundsdorfer2003numerical,perthame2012growth}. 
It also includes the transport process and the heat process as special cases. 
Mathematically, the reaction-diffusion process is governed by the partial differential equation 
\begin{align}
    \frac{\partial}{\partial t} v(t,x) &= \sum_{j=1}^d a_j(t) \frac{\partial^2}{\partial x_j^2} v(t,x) + \sum_{j=1}^d c_j(t)\frac{\partial}{\partial x_j} v(t,x) + f(t,x), \quad t \in [0,T], x \in [0,1]^d, \label{eqn:app_reaction_diffusion}\\
    v(0,x) &= v_0(x), \quad x \in [0,1]^d, \\
    v(t,x) &= 0, \quad t \in [0,T], x \in \partial [0,1]^d. 
\end{align}
Here $a_j(t)$, $c_j(t)$ and $f(t,x)$ are real-valued functions. 
Furthermore, we assume that all $a_j(t)$'s are non-negative functions and there exists an index $j_0$ such that $a_{j_0}(t) \geq a_* > 0$ uniformly in $t$. 

To solve Equation~\eqref{eqn:app_reaction_diffusion}, we first discretize the spatial variable $x$ by finite difference method and then apply quantum algorithms to solve the resulting ODE. 
Let $N$ be an integer and we consider the grid $\{(k_1/N,k_2/N,\cdots,k_d/N)\}_{1 \leq k_j \leq N-1 }$. 
Then we can approximate the spatial derivatives as 
\begin{align}
    \frac{\partial^2}{\partial x_j^2} v(t,x_1,\cdots,x_d) &\approx N^2 (v(t,\cdots,x_j+1/N,\cdots) - 2v(t,\cdots,x_j,\cdots) + v(t,\cdots,x_j-1/N,\cdots)), \\
    \frac{\partial}{\partial x_j} v(t,x_1,\cdots,x_d) &\approx 2N (v(t,\cdots,x_j+1/N,\cdots) - v(t,\cdots,x_j-1/N,\cdots)). 
\end{align}
Let $u(t)$ be an $(N-1)^d$-dimensional vector supposed to approximate $v(t,k_1/N,\cdots,k_d/N)$ by its corresponding entry, then the semi-discretized ODE form of Equation~\eqref{eqn:app_reaction_diffusion} can be written as 
\begin{equation}\label{eqn:app_reaction_diffusion_ODE}
    \frac{d}{d t} u(t) = L(t) u(t) + i H(t) v(t) + f(t), 
\end{equation}
where
\begin{equation}
    L(t) = N^2 \sum_{j=1}^d a_j(t) I^{\otimes (j-1)} \otimes D \otimes I^{\otimes (d-j) },  \quad D = \left( \begin{array}{ccccc}
        -2 &1 & & & \\
         1& -2& 1& & \\
         & \ddots& \ddots & \ddots & \\
         & & 1& -2& 1\\
         & & &1 &-2 
    \end{array} \right), 
\end{equation}
\begin{equation}
    H(t) = - \frac{1}{2} i N \sum_{j=1}^d c_j(t) I^{\otimes (j-1)} \otimes G \otimes I^{\otimes (d-j) },  \quad G = \left( \begin{array}{ccccc}
         0&1 & & & \\
         -1& 0& 1& & \\
         & \ddots& \ddots & \ddots & \\
         & & -1&0 & 1\\
         & & &-1 & 0
    \end{array} \right), 
\end{equation}
and 
\begin{equation}
    f(t) = ( f(t,k_1/N,\cdots,k_d/N) )_{1 \leq k_j \leq N-1}. 
\end{equation}

In Equation~\eqref{eqn:app_reaction_diffusion_ODE}, it is clear that both $L(t)$ and $H(t)$ are Hermitian matrices. 
Furthermore, since the eigenvalues of $D$ are given as $\lambda = - 4 \left( \sin\left( \frac{k\pi}{2N} \right) \right)^2, 1 \leq k \leq N-1$, we have 
\begin{equation*}
    L(t) \leq N^2 \sum_{j=1}^d a(t) \left( - 4 \left( \sin\left( \frac{\pi}{2N} \right) \right)^2 \right) \leq -\pi^2 a_* .  
\end{equation*}
Therefore, Equation~\eqref{eqn:app_reaction_diffusion_ODE} is a dissipative ODE with $\eta = \pi^2 a_* $, and we can apply our algorithms to prepare the history state and the final state of Equation~\eqref{eqn:app_reaction_diffusion_ODE} with a query complexity scales as $\mathcal{O}(\log^3(1/\epsilon))$, which is independent of time $T$. 

\end{document}